\documentclass[a4paper,12pt]{article}

\usepackage{amsthm}

\usepackage{amssymb}

\usepackage{fullpage}
\usepackage{amsfonts}
\usepackage{amsmath}
\usepackage{lscape}
\usepackage[round]{natbib}

\usepackage[pdftex]{graphicx}
\DeclareGraphicsExtensions{.jpg,.pdf,.png}
\graphicspath{ {../Images/}}

\newtheorem{prop}{Property}
\newtheorem{coro}{Corollary}
\newtheorem{conj}{Conjecture}

\newtheorem{Def}{Definition}

\newcommand{\revdots}{\reflectbox{$\ddots$}}
\newcommand{\dfdx}[2]{ \frac{\partial #1}{\partial #2} }

\usepackage[a4paper,pagebackref,hyperindex=true]{hyperref}
\usepackage{makeidx}
\usepackage{color}
\citeindextrue
\makeindex

\definecolor{linkcol}{rgb}{1,0,0} 
\definecolor{citecol}{rgb}{0,0,1} 

\hypersetup{
bookmarksopen=true, 
bookmarks=true,                         
bookmarksnumbered=true,    
pdftitle={A new GLM framework for analysing categorical data. Application to plant and structure development}, 
pdfauthor={Jean Peyhardi}, 
pdfsubject={PDE}, 
pdfkeywords={PDE,}  
pdftoolbar=true, 
pdfmenubar=true, 
pdfhighlight=/O, 
colorlinks=true, 
pdfpagemode=None, 
pdfpagelayout=SinglePage, 
pdffitwindow=true, 
linkcolor=linkcol, 
citecolor=citecol, 
urlcolor=linkcol, 
breaklinks= true, 
hyperindex=true,                        
}

\begin{document}

\noindent \textbf{{\LARGE A new specification of generalized linear \\ models for categorical data}}

\vspace*{1cm}

\noindent \textbf{Jean Peyhardi$^{1,3}$, Catherine Trottier$^{1,2}$, Yann Gu\'edon$^{3}$}

\vspace*{0.2cm}

{\it
\noindent $^{1}$ UM2, Institut de Math\'ematiques et Mod\'elisation de Montpellier.
 
\noindent $^{2}$ UM3, Institut de Math\'ematiques et Mod\'elisation de Montpellier.

\noindent $^{3}$ CIRAD, AGAP et Inria, Virtual Plants, 34095 Montpellier.
}

\vspace*{0.2cm}

\noindent E-mail for correspondence: \textit{jean.peyhardi@univ-montp2.fr}

\paragraph*{Abstract}
Many regression models for categorical data have been introduced in various applied fields, motivated by different paradigms. But these models are difficult to compare because their specifications are not homogeneous. The first contribution of this paper is to unify the specification of regression models for categorical response variables, whether nominal or ordinal. This unification is based on a decomposition of the link function into an inverse continuous cdf and a ratio of probabilities. This allows us to define the new family of reference models for nominal data, comparable to the adjacent, cumulative and sequential families of models for ordinal data. We introduce the notion of reversible models for ordinal data that enables to distinguish adjacent and cumulative models from sequential ones. Invariances under permutations of categories are then studied for each family. The combination of the proposed specification with the definition of reference and reversible models and the various invariance properties leads to an in-depth renewal of our view of regression models for categorical data. Finally, a family of new supervised classifiers is tested on three benchmark datasets and a biological dataset is investigated with the objective of recovering the order among categories with only partial ordering information.


\section{Introduction}
Since generalized linear models (GLMs) were first introduced by \citet{nelder72}, many regression models for categorical data have been introduced into various applied fields such as medicine, econometrics and psychology. They have been motivated by different paradigms and their specifications are not homogeneous. Several have been defined for ordinal data \citep{agresti_ordinal}, whereas only one has been defined for nominal data: the multinomial logit model introduced by \citet{luce59} (also referred to as the baseline-category logit model \citep{agresti}). The three classical models for ordinal data are the odds proportional logit model \citep{mccullagh80}, the sequential logit model \citep{tutz1990sequential_item} (also referred to as the continuation ratio logit model \citep{dobson}), and the adjacent logit model \citep{masters1982rasch,agresti_ordinal}. They have been extended by either replacing the logistic cumulative distribution function (cdf) by other cdfs (e.g. normal or Gumbel cdfs; see the grouped Cox model for instance, also referred to as the proportional hazard model), or introducing different parametrizations of the linear predictor (i.e. changing the design matrix $Z$). No such developments have been undertaken for the multinomial logit model, and one of our goals is to fill this gap.

It should be noted that the three previously mentioned models for ordinal data, and the multinomial logit model, all rely on the logistic cdf. The difference between them stems from another element in the link function. In fact, the four corresponding link functions can be decomposed into the logistic cdf and a ratio of probabilities $r$. For the odds proportional logit model, the ratio corresponds to the cumulative probabilities $P(Y\leq j)$. Here, we propose to decompose the link function of any GLM for categorical data into a cdf $F$ and a ratio $r$. Using this decomposition it can be shown that all the models for ordinal data were defined by fixing the ratio $r$ and changing the cdf $F$ and/or the design matrix $Z$. For example, all the cumulative models were obtained by fixing the cumulative probabilities $P(Y\leq j)$ as the common structure. In the same way, the two families of adjacent and sequential models were defined with probability ratios $P(Y=j|j\leq Y \leq j+1)$ and $P(Y=j|Y\geq j)$. In the same spirit the multinomial logit model can be extended by fixing only its ratio $r$ and leaving $F$ and $Z$ unrestricted. 

Our first contribution in section \ref{unification} is to unify all these models by introducing the new $(r,F,Z)$-triplet specification. In this framework, the multinomial logit model is extended, replacing the logistic cdf by other cdfs. We thus obtain a new family of models for nominal data, comparable to the three classical families of models used for ordinal data. We can now compare all the models according to the three components: ratio $r$ for structure, cdf $F$ for fit, and design matrix $Z$ for parametrization. 

Section \ref{equivalence} investigates this comparison in depth by studying the equivalences between models of different families. We revisit three equivalences between models with different ratios, shown by \citet{laara}, \citet{tutz1991sequential} and \citet{agresti_ordinal} then we generalize two of them to obtain equalities between subfamilies of models. Section \ref{invariance} investigates equivalence between models of the same family with different permutations of the categories. Some invariance properties under permutations are given for each family of models and the notion of reversible models for ordinal data that enables to distinguish adjacent and cumulative models from sequential models is introduced.

Using the extended family of models for nominal data, and their invariance property, a new family of supervised classifiers is introduced in section \ref{applications}. Theses classifiers are tested on three benchmark datasets and compared to the logistic regression. Finally focus is made on a pear tree dataset with only a partial information about ordering assumption. We recover total order among categories using sequential models and their invariance properties. 

\section{Exponential form of the categorical distribution}
The exponential family of distributions is first introduced in the multivariate case. Consider a random vector $\boldsymbol{Y}$ of $\mathbb{R}^K$ whose distribution depends on a parameter $\boldsymbol{\theta} \in \mathbb{R}^K$. The distribution belongs to the exponential family if its density can be written as
\[ f(\boldsymbol{y};\boldsymbol{\theta},\lambda) = \exp \left\lbrace \frac{\boldsymbol{y}^t \theta - b(\boldsymbol{\theta})}{\lambda} + c(\boldsymbol{y},\lambda) \right\rbrace, \label{expo} \]
where $b$, $c$ are known functions, $\lambda$ is the nuisance parameter and $\boldsymbol{\theta}$ is the natural parameter.
%

	\subsection*{Truncated multinomial distribution}
	Let $J \geq 2$ denote the number of categories for the variable of interest and $n \geq 1$ the number of trials. Let $\pi_1,\ldots,\pi_J$ denote the probabilities of each category, such that $\sum_{j=1}^J \pi_j=1$. Let the discrete vector $\boldsymbol{\tilde{Y}}$ follow the multinomial distribution
	\[ \boldsymbol{\tilde{Y}} \sim \mathcal{M}(n,(\pi_1,\ldots,\pi_J)), \]
with $\sum_{j=1}^J \tilde{y}_j=n$. It should be remarked that the multinomial distribution is not exactly a generalization of the binomial distribution, just looking at the dimension of $\boldsymbol{\tilde{Y}}$. In fact, the constraint $\sum_{j=1}^J \pi_j=1$ expresses one of the probabilities in terms of the others. By convention we choose to put the last category aside: $\pi_J=1-\sum_{j=1}^{J-1} \pi_j$. Finally, we define the truncated multinomial distribution 
	\[ \boldsymbol{Y} \sim t\mathcal{M}(n,(\pi_1,\ldots,\pi_{J-1})), \]
where $\boldsymbol{Y}$ is the truncated vector of dimension $J-1$ with the constraint $0 \leq \sum_{j=1}^{J-1} y_j \leq n$. The probabilities $\pi_j$ are strictly positive and $\sum_{j=1}^{J-1} \pi_j < 1$ to avoid degenerate cases. Let $\boldsymbol{\pi}$ denote the truncated vector $(\pi_1,\ldots,\pi_{J-1})^t$ with $E(\boldsymbol{Y})=n\boldsymbol{\pi}$. For $J=2$ the truncated multinomial distribution is the Bernoulli distribution if $n=1$ and the binomial distribution if $n > 1$. In the GLM framework, only $\pi$ is related to the explanatory variables thus we focus on the case $n=1$. One observation $y$ is an indicator vector of the observed category (the null vector corresponding to the last category). The truncated multinomial distribution can be written as follows

\begin{equation*}
f(\boldsymbol{y};\boldsymbol{\theta}) = \exp \lbrace \boldsymbol{y}^t \theta - b(\boldsymbol{\theta}) \rbrace,
\end{equation*}
where
\begin{equation*}
\theta_j = \ln \left\lbrace \frac{\pi_j}{1- \sum_{j=1}^{J-1}\pi_k } \right\rbrace
\end{equation*}
for $j=1,\ldots,J-1$, and
\begin{equation*}
b(\boldsymbol{\theta}) = \ln \left\lbrace 1+ \sum_{j=1}^{J-1} \exp (\theta_j) \right\rbrace. 
\end{equation*}
Using the nuisance parameter $\lambda=1$ and the null function $c(y,\lambda)=0$, we see that the truncated multinomial distribution $t\mathcal{M}(\pi)$ belongs to the exponential family of dimension $K=\dim(\boldsymbol{Y})=\dim(\boldsymbol{\theta})=\dim(\boldsymbol{\pi})=J-1$.


\section{\label{unification}Specification of generalized linear models for categorical data}
Consider a regression analysis, with the multivariate response variable $Y$ and the vector of $Q$ explanatory variables $\boldsymbol{X}=(X_1, \ldots, X_Q)$ in a general form (i.e. categorical variables being represented by indicator vectors). The dimension of $\boldsymbol{X}$ is thus denoted by $p$ with $p \geq Q$. We are interested in the conditional distribution of $\boldsymbol{Y}|\boldsymbol{X}$, observing the values $(\boldsymbol{y}_i,\boldsymbol{x}_i)_{i=1,\ldots,n}$ taken by the pair $(\boldsymbol{Y},\boldsymbol{X})$. All the response variables $\boldsymbol{Y}_i$ are assumed to be conditionally independent of each other, given $\{\boldsymbol{X}_i=\boldsymbol{x}_i\}$. The variables $\boldsymbol{Y}_i$ follow the conditional truncated multinomial distribution
	\[ \boldsymbol{Y}_i|\boldsymbol{X}_i=\boldsymbol{x}_i \sim t\mathcal{M}(\boldsymbol{\pi}(\boldsymbol{x}_i)), \]
with at least $J=2$. In the following we will misuse some notations for convenience. For example, we will often forget the conditioning on $\boldsymbol{X}$ and the individual subscript $i$. Moreover, the response variable will sometimes be considered as a univariate categorical variable $Y \in \{1,\ldots,J\}$ in order to use the univariate notation $\{Y=j\}$ instead of the multivariate notation $\lbrace Y_1=0, \; \ldots,Y_{j-1}=0, \; Y_j=1, \;Y_{j+1}=0, \; \ldots,\; Y_{J-1}=0 \rbrace$.

	\subsection{Decomposition of the link function}
The definition of a GLM for categorical data includes the specification of a link function $g$ which is a diffeomorphism from $ \mathcal{M}= \lbrace \boldsymbol{\pi} \in \; ]0,1[^{J-1} \vert \sum_{j=1}^{J-1} \pi_j < 1\rbrace$ to an open subset $\mathcal{S}$ of $\mathbb{R}^{J-1}$, between the expectation $\boldsymbol{\pi}=E[\boldsymbol{Y}|\boldsymbol{X}=\boldsymbol{x}]$ and the linear predictor $\boldsymbol{\eta}=(\eta_1,...,\eta_{J-1})^t$. It also includes the parametrization of the linear predictor $\boldsymbol{\eta}$ which can be written as the product of the design matrix $Z$ (as a function of $\boldsymbol{x}$) and the vector of parameters $\boldsymbol{\beta}$ \citep{tutz_book}. Given the vector of explanatory variables $\boldsymbol{x}$, a GLM for a categorical response variable is characterized by the equation $ g(\boldsymbol{\pi}) = Z \boldsymbol{\beta} $, where there are exactly $J-1$ equations $ g_j(\boldsymbol{\pi})=\eta_j$. The model can also be represented by the following diagram
\begin{center}
\begin{tabular}{ccccc}
 & $Z$ & & $g^{-1}$ & \\
 $\mathcal{X}$ & $\longrightarrow$ & $\mathcal{S}$ & $\longrightarrow$ & $\mathcal{M}$, \\
 $\boldsymbol{x}$ & $\longmapsto$ & $\boldsymbol{\eta}$ & $\longmapsto$ & $\boldsymbol{\pi}$, \\
\end{tabular}
\end{center}
where $\mathcal{X}$ is the space of explanatory variables of dimension $p$. 

All the classical link functions (see \citet{agresti,tutz2011}) have the same structure which we propose to write as
\begin{align}
g_j = F^{-1} \circ r_j , \;\; j=1,\ldots,J-1 , \label{structure}
\end{align}
where $F$ is a continuous and strictly increasing cumulative distribution function and $\boldsymbol{r}=(r_1,\ldots,r_{J-1})^t$ is a diffeomorphism from $\mathcal{M}$ to an open subset $\mathcal{P}$ of $]0,1[^{J-1}$. Finally, given $\boldsymbol{x}$, we propose to summarize a GLM for categorical response variable by
\[ \boldsymbol{r}(\boldsymbol{\pi})=\boldsymbol{\mathcal{F}}(Z \boldsymbol{\beta}), \]
where $\boldsymbol{\mathcal{F}}(\boldsymbol{\eta}) = (F(\eta_1),\ldots,F(\eta_{J-1}))^t$. The model can thus be presented by the following diagram
\begin{center}
\begin{tabular}{ccccccc}
 & $Z$ & & $\mathcal{F}$ & & $r^{-1}$  \\
 $\mathcal{X}$ & $\longrightarrow$ & $\mathcal{S}$ & $\longrightarrow$ &  $\mathcal{P}$& $\longrightarrow$ & $\mathcal{M}$, \\
 $\boldsymbol{x}$ & $\longmapsto$ & $\boldsymbol{\eta}$ & $\longmapsto$ & $\boldsymbol{p}$  & $\longmapsto$ & $\boldsymbol{\pi}$. \\
\end{tabular}
\end{center}

	\subsection{\textit{(r,F,Z)} specification of GLMs for categorical data}
	In the following, we will describe in more detail the components $Z$, $F$, $r$ and their modalities.
	
		\paragraph{Design matrix $Z$:}
		Each linear predictor has the form $\eta_j=\alpha_j + x^t\delta_j$ with $ \boldsymbol{\beta} = (\alpha_1,\ldots,\alpha_{J-1},$ $ \boldsymbol{\delta_1}^t, \ldots, \boldsymbol{\delta_{J-1}}^t) \in \mathbb{R}^{(J-1)(1+p)} $. In general, the model is defined without constraints, like for the multinomial logit model. But linear equality constraints, called contrasts, can be added between different slopes $\boldsymbol{\delta_j}$, for instance. The most common constraint is the equality of all slopes, like for the odds proportional logit model. The corresponding constrained space $\mathcal{C} = \lbrace \boldsymbol{\beta} \in \mathbb{R}^{(J-1)(1+p)} | \boldsymbol{\delta_1}=\ldots=\boldsymbol{\delta_{J-1}} \rbrace $ may be identified to $\tilde{\mathcal{C}} = \mathbb{R}^{(J-1)+p} $. Finally, the constrained space is represented by a design matrix, containing the vector of explanatory variables $\boldsymbol{x}$. For example, the complete design $(J-1) \times (J-1)(1+p)$-matrix $Z_{c}$ (without constraint) has the following form
\[ Z_{c} = \begin{pmatrix}
1 & & &\boldsymbol{x}^t& & \\
 &\ddots & & &\ddots & \\
 & & 1& & &\boldsymbol{x}^t
\end{pmatrix}.\]		
The proportional design $(J-1) \times (J-1+p)$-matrix $Z_{p}$ (common slope) has the following form
\[ Z_{p} = \begin{pmatrix}
1 & & &\boldsymbol{x}^t \\
 &\ddots & & \vdots \\
 & & 1& \boldsymbol{x}^t
\end{pmatrix}.\]
The model without slopes ($ \boldsymbol{\delta_1} = \ldots = \boldsymbol{\delta_{J-1}} = \boldsymbol{0}$), considered as the minimal response model, is defined with different intercepts $\alpha_j$ (the design matrix is the identity matrix of dimension $J-1$). In most cases, the design matrix contains the identity matrix as minimal block, such as $Z_c$ and $Z_p$. These two matrices are sufficient to define all the classical models. It should be noted that for a given constrained space $\mathcal{C}$, there are an infinity of corresponding design matrices which will be considered as equivalent. For example
\[ Z_{p}^{\prime} = \begin{pmatrix}
1 & & &-\boldsymbol{x}^t \\
\vdots &\ddots & & \vdots \\
1 &\ldots & 1& -\boldsymbol{x}^t
\end{pmatrix}\]
is equivalent to $Z_p$. In the following, the design matrices $Z_p$ and $Z_c$ are considered as the representative element of their equivalence class and the set of all possible design matrices $Z$, for a given vector of explanatory variables $\boldsymbol{x}$, will be denoted by $\mathfrak{Z}$. This set $\mathfrak{Z}$ contains all design matrices between $Z_p$ and $Z_c$, with number of columns between $J-1+p$ and $(J-1)(1+p)$.
	
		\paragraph{Cumulative distribution function $F$:} 
		The most commonly used symmetric distributions are the logistic and normal distributions, but Laplace and Student distributions may also be useful. The most commonly used asymmetric distribution is the Gumbel min distribution
\[ F(w) = 1 - \exp \left\lbrace -\exp (w) \right\rbrace.\]
Let $\tilde{F}$ denote the symmetric of $F$ (i.e. $\tilde{F}(w)=1-F(-w)$). The symmetric of the Gumbel min distribution is the Gumbel max distribution
\[ \tilde{F}(w) = \exp \left\lbrace -\exp (-w) \right\rbrace.\]
All these cdfs, being diffeomorphisms from $\mathbb{R}$ to $]0,1[$, ease the interpretation of estimated parameter $\hat{\boldsymbol{\beta}}$. The exponential distribution, which is a diffeomorphism from $\mathbb{R}_+^*$ to $]0,1[$, is also used but the positivity constraint on predictors may lead to divergence of estimates. The Pareto distribution is defined with the shape parameter $a \geq 1$ by
\[ F(w) = 1 - w^{-a}, \]
and shares with the exponential distribution constraint on the support $[1,+\infty[$. 
			As noted by \citet{tutz1991sequential}, ``distribution functions generate the same model if they are connected by a linear transformation". For instance, if the connexion is made through a location parameter $u$ and a scale parameter $s$ such that $F_{u,s}(w)=F\{(w-u)/s\}$, we have for $j=1,\ldots,J-1$
\[ F_{u,s}(\eta_j(\boldsymbol{x})) = F \left( \frac{\eta_j(\boldsymbol{x})-u}{s} \right) = F \left( \frac{\alpha_j-u}{s}+ \boldsymbol{x}^t \frac{\boldsymbol{\delta_j}}{s}\right), \]
and obtain an equivalent model using the reparametrization $\alpha^{\prime}_j= (\alpha_j-u)/s$ and $\boldsymbol{\delta_j}^{\prime}=\boldsymbol{\delta_j}/s$ for $j=1,\ldots,J-1$. This is the case for all distributions previously introduced. But Student (respectively Pareto) distributions, with different degrees of freedom (respectively parameters shape $a$), are not connected by a linear transformation. Therefore they lead to different likelihood maxima. In applications, Student distributions will be used with few degrees of freedom. Playing on the symmetrical or asymmetrical character and the more or less heavy tails of distributions may markedly improve model fit. In the following, the set of all continuous cdf $F$ (respectively continuous and symmetric cdf $F$) will be denoted by $\mathfrak{F}$ (respectively by $\tilde{\mathfrak{F}}$).		

		\paragraph{Ratio of probabilities $r$:} 
		The linear predictor $\boldsymbol{\eta}$ is not directly related to the expectation $\boldsymbol{\pi}$, through the cdf $F$, but to a particular transformation $r$ of $\boldsymbol{\pi}$ which we call the ratio. 

\vspace*{0.2cm}

\noindent In this context, the odds proportional logit model for instance relies on the cumulative ratio defined by
\[ r_j(\boldsymbol{\pi}) = \pi_1 + \ldots + \pi_j, \]
for $j=1,\ldots,J-1$. If there is a total order among categories, cumulative models can be used and interpreted by introducing a latent continuous variable $V$ having cdf $F$ \citep{mccullagh80}. The linear predictors $(\eta_j)_{j=1,\ldots,J-1}$ are then strictly ordered and we obtain for $j=1,\ldots,J-1$
\[ \{ Y \leq j \} \Leftrightarrow \{ V \leq \eta_j \}. \]

\vspace*{0.2cm}

\noindent The continuation ratio logit model relies on the sequential ratio defined by
\[ r_j(\boldsymbol{\pi}) = \frac{\pi_j}{\pi_j + \ldots + \pi_J},\]
for $j=1,\ldots,J-1$. Total order may be interpreted in a different way with sequential models. A sequential model corresponds to a sequential independent latent continuous process $(V_t)_{t=1,\ldots,J-1}$ having cdf $F$ \citep{tutz1990sequential_item}. This process is governed by
\[ \{ Y=j \} \Leftrightarrow \bigcap_{t=1}^{j-1} \{ V_t > \eta_t \} \bigcap \{ V_j \leq \eta_j \}, \]
for $j=1,\ldots,J-1$. The conditional event $\{Y=j | Y\geq j\}$ can be expressed by
\[ \{Y=j | Y\geq j\} \Leftrightarrow \{ V_j \leq \eta_j \}. \]

\vspace*{0.2cm}

\noindent The adjacent logit model is based on the adjacent ratio defined by
\[ r_j(\boldsymbol{\pi}) = \frac{\pi_j}{\pi_j + \pi_{j+1}},\]
for $j=1,\ldots,J-1$. Adjacent models are not directly interpretable using latent variables.

\vspace*{0.2cm}

\noindent Unlike these models for ordinal data, we propose to define a ratio that is independent of the category ordering assumption. Using the structure of the multinomial logit model, we define the reference ratio for each category $j=1,\ldots,J-1$ as
\[ r_j(\boldsymbol{\pi}) = \frac{\pi_j}{\pi_j + \pi_J}. \]
Each category $j$ is then related to a reference category (here $J$ by convention) and thus no category ordering is assumed. Therefore, the reference ratio allows us to define new GLMs for nominal response variables.

%

\vspace*{0.5cm}

In the following, each GLM for categorical data will be specified by a $(r,F,Z)$ triplet. Table \ref{table_rFZ_5_model} shows $(r,F,Z)$ triplet specifications for classical models. This specification enables to define an enlarged family of GLMs for nominal response variables (referred to as the reference family) using $(\textnormal{reference},F,Z)$ triplets, which includes the multinomial logit model specified by the (reference, logistic, complete) triplet. GLMs for nominal and ordinal response variables are usually defined with different design matrices $Z$; see the first two rows in table \ref{table_rFZ_5_model}. Fixing the design matrix $Z$ eases the comparison between GLMs for nominal and ordinal response variables.

\begin{table}
\caption{\label{table_rFZ_5_model} $(r,F,Z)$ specification of five classical GLMs for categorical data.}
\centering
\begin{tabular}{|c|c|}
 \hline
  & \\
\textit{Multinomial logit model} & \\
 & \\
 $ \displaystyle P(Y=j)= \frac{\exp (\alpha_j + x^T \delta_j)}{1+ \sum_{k=1}^{J-1} \exp (\alpha_k + x^T \delta_k)}$ & (reference, logistic, complete) \\
 & \\
\hline
  & \\
\textit{Odds proportional logit model} & \\
  & \\
$ \displaystyle \ln \left\lbrace \frac{P(Y\leq j)}{1-P( Y \leq j)} \right\rbrace =\alpha_j+x^T\delta$ & (cumulative, logistic, proportional) \\
   & \\
\hline
   & \\
\textit{Proportional hazard model} & \\
\textit{(Grouped Cox Model)} &  \\
&  \\
$ \displaystyle \ln \left\lbrace - \ln P(Y>j | Y\geq j) \right\rbrace =\alpha_j+x^T\delta$ & (sequential, Gumbel min, proportional) \\
  & \\
\hline
   & \\
\textit{Adjacent logit model} & \\
&  \\
$ \displaystyle \ln \left\lbrace \frac{P(Y=j)}{P(Y=j+1)} \right\rbrace =\alpha_j+x^T\delta_j$ & (adjacent, logistic, complete) \\
  & \\
\hline
   & \\
\textit{Continuation ratio logit model} & \\
&  \\
$ \displaystyle \ln \left\lbrace \frac{P(Y=j)}{P( Y > j)} \right\rbrace =\alpha_j+x^T\delta_j$ & (sequential, logistic, complete) \\
&  \\
\hline
\end{tabular}
\end{table}

	\subsection{Compatibility of the three components \textit{r}, \textit{F} and \textit{Z}}\label{compatibility}
	A GLM for categorical data is specified by an $(r,F,Z)$ triplet but is it always defined? The condition $\boldsymbol{\pi}(\boldsymbol{x}) \in \mathcal{M}$ is required for all $\boldsymbol{x} \in \mathcal{X}$. It should be noted that reference, adjacent and sequential ratios are defined with $J-1$ different conditioning. Therefore the linear predictors $\eta_j$ are not constrained one to another. Neither $\mathcal{P}$ nor $\mathcal{S}$ are constrained ($\mathcal{P}=]0,1[^{J-1}$ and $\mathcal{S}=\mathbb{R}^{J-1}$) and thus no constraint on parameter $\boldsymbol{\beta}$ is required.

	The situation is different for the cumulative ratio, because the probabilities $r_j(\boldsymbol{\pi})$ are not conditional but linked ($r_{j+1}(\boldsymbol{\pi})=r_j(\boldsymbol{\pi})+ \pi_{j+1}$). Both $\mathcal{P}$ and $\mathcal{S}$ are constrained ($\mathcal{P}= \lbrace \boldsymbol{r} \in ]0,1[^{J-1} | r_1<\ldots<r_{J-1} \rbrace$ and $\mathcal{S}=\lbrace \boldsymbol{\eta} \in \mathbb{R}^{J-1} | \eta_1<\ldots<\eta_{J-1} \rbrace$). Therefore the definition of a cumulative model entails constraints on $\boldsymbol{\beta}=(\alpha_1,\ldots,\alpha_{J-1},\boldsymbol{\delta_1}^t,\ldots,\boldsymbol{\delta_{J-1}}^t)$. Without loss of generality, we will work hereinafter with only one explanatory variable $x \in \mathcal{X}$. The constraints are different depending on the form of $\mathcal{X}$.
	
			\paragraph*{Case 1: $x$ is categorical} 
			then $\mathcal{X}= \lbrace \boldsymbol{x} \in \{0,1\}^{C-1} | \; \sum_{c=1}^{C-1} x_c \in \{0,1\} \rbrace$. In this case, the form of the linear predictors is
\[ \eta_j(x) = \alpha_j + \sum_{c=1}^{C-1} \textbf{1}_{\lbrace X=c \rbrace} \; \delta_{j,c},\]
and the constraints $\eta_j(x) < \eta_{j+1}(x) \; \forall x \in \mathcal{X}$ are equivalent to
\[ \left\{
\begin{array}{rcll}
\alpha_j   &  <   &  \alpha_{j+1},  & \\
\delta_{j,c} - \delta_{j+1,c} & < & \alpha_{j+1} -\alpha_j  , & \forall c \in \lbrace1,\ldots,C-1 \rbrace.
\end{array}
\right. \]
A sufficient condition is 
\[ \left\{
\begin{array}{llll}
\alpha_j   &  <   &  \alpha_{j+1},  & \\
\delta_{j,c} & \leq & \delta_{j+1,c}, & \forall c \in \lbrace1,\ldots,C-1 \rbrace.
\end{array}
\right. \]

			\paragraph*{Case 2: $x$ is continuous} 
			then $\mathcal{X} \subseteq \mathbb{R}$. In this case, the form of the linear predictors is
\[ \eta_j(x) = \alpha_j + \delta_j x.\]
Since the $\eta_j$ must be ordered on $\mathcal{X}$, three domains of definition $\mathcal{X}$ must be differentiated:

\noindent \underline{$\mathcal{X} = \mathbb{R}$:}
				$\eta_j$ are ordered and parallel straight lines
\[ \left\{
\begin{array}{lll}
\alpha_j   &  <   &  \alpha_{j+1},  \\
\delta_j & = & \delta_{j+1}.
\end{array}
\right. \]
This is the case of the odds proportional logit model.

\noindent \underline{$\mathcal{X} = \mathbb{R_+}$:}
				$\eta_j$ are ordered and non-intersected half-lines
\[ \left\{
\begin{array}{lll}
\alpha_j   &  <   &  \alpha_{j+1},  \\
\delta_j & \leq & \delta_{j+1}.
\end{array}
\right. \]
This is the case of a positive continuous variable $X$, such as a size or a dosage for instance. Moreover, if $X$ is strictly positive, the intercepts $\alpha_j$ can be common.

\vspace*{0.1cm}

\noindent \underline{$\mathcal{X} = [a,b]$:}
				$\eta_j$ are ordered and non-intersected segments. The constraints cannot be simply rewritten in terms of intercept and slope constraints. For the last two cases a vector of probabilities $\pi(x)$ for $x$ out of $\mathcal{X}$ cannot always be predicted (see figure \ref{predictors_constraints}).
		
		 Finally, cumulative models are not always defined and thus some problems may occur in the application of Fisher's scoring algorithm.
		 		
\begin{figure}
\includegraphics[width=1\textwidth]{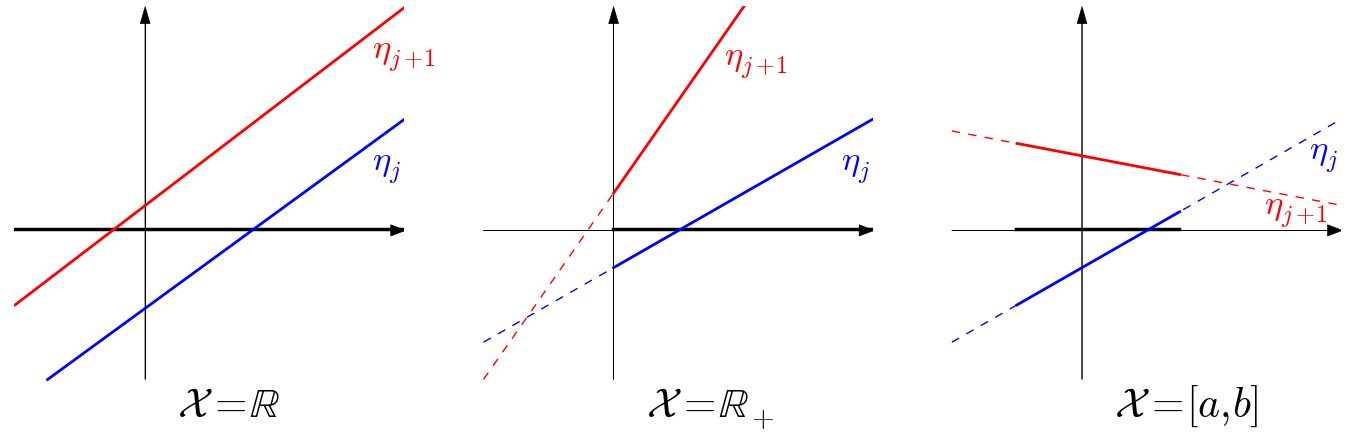}
\caption{\label{predictors_constraints} Linear predictors for different configurations of the continuous space $\mathcal{X}$.}
\end{figure}

		\subsection{Fisher's scoring algorithm}
For maximum likelihood estimation, the iteration of Fisher's scoring algorithm is given by
 \[ \boldsymbol{\beta}^{[t+1]} = \boldsymbol{\beta}^{[t]} - \left\lbrace \mbox{E} \left(\frac{\partial^2 l}{\partial \boldsymbol{\beta}^t \partial \boldsymbol{\beta}} \right)_{\boldsymbol{\beta}=\boldsymbol{\beta}^{[t]}} \right\rbrace^{-1} \left( \frac{\partial l}{\partial \boldsymbol{\beta}} \right)_{\boldsymbol{\beta}=\boldsymbol{\beta}^{[t]}}. \]
For the sake of simplicity, the algorithm is detailed for only one observation $(\boldsymbol{y},\boldsymbol{x})$ with $l=\ln P(\boldsymbol{Y}=\boldsymbol{y}|\boldsymbol{X}=\boldsymbol{x};\boldsymbol{\beta})$. Using the chain rule we obtain the score
\[ \frac{\partial l}{\partial \boldsymbol{\beta}} = \frac{\partial \boldsymbol{\eta}}{\partial \boldsymbol{\beta}} \; \frac{\partial \boldsymbol{\pi}}{\partial \boldsymbol{\eta}} \; \frac{\partial \boldsymbol{\theta}}{\partial \boldsymbol{\pi}} \; \frac{\partial l}{\partial \boldsymbol{\theta}}. \]
Since the response distribution belongs to the exponential family, the score becomes
\[ \frac{\partial l}{\partial \boldsymbol{\beta}} = Z^t \; \frac{\partial \boldsymbol{\pi}}{\partial \boldsymbol{\eta}} \; \mbox{Cov}(\boldsymbol{Y}|\boldsymbol{x})^{-1} \; (\boldsymbol{y}-\boldsymbol{\pi}). \]
Then using decomposition (\ref{structure}) of the link function we obtain
\begin{align}
 \frac{\partial l}{\partial \boldsymbol{\beta}} & = Z^t \; \frac{\partial \boldsymbol{\mathcal{F}}}{\partial \boldsymbol{\eta}} \; \frac{\partial \boldsymbol{\pi}}{\partial \boldsymbol{r}} \; \mbox{Cov}(\boldsymbol{Y}|\boldsymbol{x})^{-1} \; (\boldsymbol{y}-\boldsymbol{\pi}). \label{score}
\end{align}
Again using property of the exponential family and decomposition (\ref{structure}) of the link function, we obtain Fisher's information matrix
\begin{align}
\textnormal{E} \left(\frac{\partial^2 l}{\partial \boldsymbol{\beta}^t \partial \boldsymbol{\beta}} \right) & = - \frac{\partial \boldsymbol{\pi}}{\partial \boldsymbol{\beta}} \; \mbox{Cov}(\boldsymbol{Y}|\boldsymbol{x})^{-1} \; \frac{\partial \boldsymbol{\pi}}{\partial \boldsymbol{\beta}^t} \nonumber \\
 & = - Z^t \; \frac{\partial \boldsymbol{\pi}}{\partial \boldsymbol{\eta}} \; \mbox{Cov}(\boldsymbol{Y}|\boldsymbol{x})^{-1} \; \frac{\partial \boldsymbol{\pi}}{\partial \boldsymbol{\eta}^t} \; Z \nonumber \\
\textnormal{E} \left(\frac{\partial^2 l}{\partial \boldsymbol{\beta}^t \partial \boldsymbol{\beta}} \right) & = - Z^t \; \frac{\partial \boldsymbol{\mathcal{F}}}{\partial \boldsymbol{\eta}} \; \frac{\partial \boldsymbol{\pi}}{\partial \boldsymbol{r}} \; \mbox{Cov}(\boldsymbol{Y}|\boldsymbol{x})^{-1} \; \frac{\partial \boldsymbol{\pi}}{\partial \boldsymbol{r}^t} \; \frac{\partial \boldsymbol{\mathcal{F}}}{\partial \boldsymbol{\eta}^t} \; Z . \label{information}
\end{align}
We only need to evaluate the associated density function $\{f(\eta_j)\}_{j=1,\ldots,J-1}$ to compute the diagonal Jacobian matrix $\partial \boldsymbol{\mathcal{F}} / \partial \boldsymbol{\eta}$. For details on computation of the Jacobian matrix $\partial \boldsymbol{\pi} / \partial \boldsymbol{r}$ according to each ratio, see appendix \ref{appendix_dpi_dr}. Generic decomposition \eqref{score} of score and decomposition \eqref{information} of Fisher's information matrix according to $(r,F,Z)$ triplet eases implementation of supplementary $F$ and $Z$ in this framework.

\section{\label{equivalence} Equalities and equivalences between models}

In the previous section, we have shown that any $(r,F,Z)$ triplet define a particular GLM for categorical data. But some triplets may define the same model. This section focuses on equivalences between GLMs for categorical data. All the following properties strongly depend on the link function, especially on the ratio. It should be noted that for $J=2$, the four ratios are the same, leading to the Bernoulli case. Hence we only focus on $J>2$. The terminology ``model" is used when the three components $r$, $F$ and $Z$ are instantiated (even if the parameter $\boldsymbol{\beta}$ is unknown), whereas the terminology ``family of models" is used when only the ratio $r$ is instantiated. The four ratios, corresponding to different structures, define four families of models. The family of reference models, for instance, is defined by \{(reference, $F$, $Z$)$; \; F \in \mathfrak{F}$, $Z \in \mathfrak{Z}$\}.

The truncated multinomial distribution $t\mathcal{M}(\boldsymbol{\pi})$ being specified by the parameter $\boldsymbol{\pi}$ of dimension $J-1$, the distribution of $\boldsymbol{Y}|\boldsymbol{X}=\boldsymbol{x}$ is then fully specified by the $(r,F,Z)$ triplet for a fixed value of $\boldsymbol{\beta} \in \tilde{\mathcal{C}}$ since
\[ \boldsymbol{\pi}= r^{-1} \circ \mathcal{F} \{ Z(x) \boldsymbol{\beta} \}. \]
Equality and equivalence between two models are defined using the $(r,F,Z)$ specification. 
\begin{Def}
Two models $(r,F,Z)$ and $(r^{\prime},F^{\prime},Z^{\prime})$ are said to be equal if the corresponding distributions of $\boldsymbol{Y}|\boldsymbol{X}=\boldsymbol{x}$ are equal for all $\boldsymbol{x}$ and all $\boldsymbol{\beta}$
\[ r^{-1} \circ \mathcal{F} \{ Z(\boldsymbol{x}) \boldsymbol{\beta} \} = r^{\prime-1} \circ \mathcal{F}^{\prime} \{ Z^{\prime}(\boldsymbol{x}) \boldsymbol{\beta} \}, \; \forall \boldsymbol{x} \in \mathcal{X},\; \forall \boldsymbol{\beta} \in \tilde{\mathcal{C}}. \]
\end{Def}

\begin{Def}
Two models $(r,F,Z)$ and $(r^{\prime},F^{\prime},Z^{\prime})$ are said to be equivalent if one is a reparametrization of the other, and conversely. Hence, there exists a bijection $h$ from $\tilde{\mathcal{C}}$ to $\tilde{\mathcal{C}}^{\prime}$ such that
\[ r^{-1} \circ \mathcal{F} \{ Z(\boldsymbol{x}) \boldsymbol{\beta} \} = r^{\prime-1} \circ \mathcal{F}^{\prime} \{ Z^{\prime}(\boldsymbol{x}) h(\boldsymbol{\beta}) \}, \; \forall \boldsymbol{x} \in \mathcal{X},\; \forall \boldsymbol{\beta} \in \tilde{\mathcal{C}}. \]
\end{Def}
\noindent Let us remark that equality implies equivalence.

		\subsection{Comparison between cumulative and sequential models}		
		Two equivalences between cumulative and sequential models have already been shown in the literature: a first one for the Gumbel min cdf, and another one for the exponential cdf. We extend the second one in the context of equality of models for other design matrices.
		
			\subsubsection{Gumbel min cdf}
		 	\citet{laara} showed the equivalence
\begin{equation}
\mbox{(sequential, Gumbel min, proportional)} \Leftrightarrow \mbox{(cumulative, Gumbel min, proportional)}.\label{eq_laara}
\end{equation}
In this context, \citet{tutz1991sequential} noted that ``the equivalence holds only for the simple version of the models where the thresholds are not determined by the explanatory variables". In our framework, this means that (sequential, Gumbel min, $Z$) and (cumulative, Gumbel min, $Z$) models are equivalent only for the proportional design matrix $Z_p$. Consider the bijection $H$ between the two predictor spaces $\mathcal{S}= \mathbb{R}^{J-1}$ and $\mathcal{S}^{\prime}=\lbrace \boldsymbol{\eta}^{\prime} \in \mathbb{R}^{J-1} | \eta_1^{\prime}<\ldots<\eta_{J-1}^{\prime} \rbrace$ defined by
\begin{equation}
 \eta_j^{\prime} = \displaystyle \ln \left\lbrace \sum_{k=1}^j \exp (\eta_k) \right\rbrace, \label{reparametrization_laara}
 \end{equation}
for $j=1,\ldots,J-1$. To be a linear predictor, $\boldsymbol{\eta}^{\prime}$ must be linear with respect to $\boldsymbol{x}$. Rewriting $\eta^{\prime}_j$ as
\[ \eta_j^{\prime} = \displaystyle \ln \left\lbrace \sum_{k=1}^j \exp (\alpha_k) \exp (\boldsymbol{x}^t\boldsymbol{\delta_k}) \right\rbrace, \]
we see that linearity with respect to $\boldsymbol{x}$ holds if and only if $\boldsymbol{\delta_1}=\ldots=\boldsymbol{\delta_{J-1}}$. This corresponds to the proportional design matrix. Finally the equivalence of \citet{laara} holds with the bijection $h$ between $\tilde{\mathcal{C}} = \mathbb{R}^{J-1+p} $ and $\tilde{\mathcal{C}}^{\prime} = \lbrace \boldsymbol{\beta}\in \mathbb{R}^{J-1+p} | \alpha_1 < \ldots < \alpha_{J-1} \rbrace $ defined by 
\[ \left\lbrace
\begin{array}{lcl}
 \alpha_j^{\prime} &=& \displaystyle \ln \left\lbrace \sum_{k=1}^j \exp (\alpha_k) \right\rbrace, \; \mbox{for } j=1,\ldots,J-1, \\
\boldsymbol{\delta}^{\prime} &=& \boldsymbol{\delta}. \\
\end{array}\right.\]

			\subsubsection{Pareto cdfs}
			For all Pareto distributions we show an equivalence between cumulative and sequential models for a particular design matrix.
\begin{prop}
Let $\mathcal{P}_a$ denote the Pareto cdf with shape parameter $a\geq 1$ and $Z_0$ the $(J-1) \times (J-1+p)$-design matrix
\[ Z_0 = \begin{pmatrix}
1 &  & &  & 0 \\
  & \ddots  &  & & \vdots  \\
 & &  1& & 0\\
 & &  & 1 & \boldsymbol{x}^t
\end{pmatrix}.
\]
Then we have
\begin{equation}
 \mbox{(cumulative, }\mathcal{P}_a, Z_0) \Leftrightarrow \mbox{(sequential, }\mathcal{P}_a, Z_0). \label{eq_Pareto}
\end{equation}
\end{prop}
\begin{proof}	
Assume that the distribution of $\boldsymbol{Y}|\boldsymbol{X}=\boldsymbol{x}$ is defined by the (cumulative, exponential, $Z$) model with an unknown design matrix $Z \in \mathfrak{Z}$.
For $j=1,\ldots,J-1$ we obtain
\begin{equation}
\pi_1+\ldots+\pi_j = 1 - \eta_j^{-a}. \label{cum_Pareto}
\end{equation}
The sequential ratio can be rewritten in terms of cumulative ratio
\begin{equation}
\frac{\pi_j}{\pi_j+\ldots+\pi_J} = \frac{(\pi_1+\ldots+\pi_j )-(\pi_1+\ldots+\pi_{j-1})}{1-(\pi_1+\ldots+\pi_{j-1})}, \label{cum_seq}
\end{equation}
for $j=2,\ldots,J-1$. Using (\ref{cum_Pareto}) it becomes
\[ \frac{\pi_j}{\pi_j+\ldots+\pi_J} = 1- \left( \frac{\eta_j}{\eta_{j-1}} \right)^{-a},\]
for $j=2,\ldots,J-1$. Consider the reparametrization
\begin{equation}
 \eta_j^{\prime} = \left\lbrace
\begin{array}{cl}
\displaystyle \eta_j & \mbox{for } j=1, \\
\displaystyle \frac{\eta_j}{\eta_{j-1}}, & \mbox{for } j=2,\ldots,J-1.
\end{array}\right. \label{reparametrization_Pareto}
\end{equation}
$\boldsymbol{\eta}^{\prime}$ is linear with respect to $\boldsymbol{x}$ if and only if $\boldsymbol{\delta_{j-1}}=0$ for $j=2,\ldots,J-1$, i.e. if and only if $Z=Z_0$. 
\end{proof}	
	
			\subsubsection{Exponential cdf}
			\citet{tutz1991sequential} showed the equivalence
\begin{equation} 
\mbox{(cumulative, exponential, complete)} \Leftrightarrow \mbox{(sequential, exponential, complete)} \label{eq_tutz}
\end{equation}
Using the reparametrization behind this equivalence, we obtain the following property:
\begin{prop}\label{equal_cum_seq}
Let $A$ be the square matrix of dimension $J-1$
\[ A = \begin{pmatrix}
1 &  & &   \\
-1  & 1  &  &   \\
 & \ddots & \ddots & \\
 & & -1 & 1
\end{pmatrix}.
\]
Then we have the equality of models
\[ \mbox{(cumulative, exponential, } Z) = \mbox{(sequential, exponential, } AZ),\]
for any design matrix $Z \in \mathfrak{Z}$.
\end{prop}
\begin{proof}
Assume that the distribution of $\boldsymbol{Y}|\boldsymbol{X}=\boldsymbol{x}$ is defined by the (cumulative, exponential, $Z$) model with an unknown design matrix $Z \in \mathfrak{Z}$.
For $j=1,\ldots,J-1$ we obtain
\begin{align}
\pi_1+\ldots+\pi_j & = 1 - \exp(-\eta_j). \label{cum_exp}
\end{align}
Using (\ref{cum_seq}) it becomes
\[ \frac{\pi_j}{\pi_j+\ldots+\pi_J} = 1- \exp\lbrace -(\eta_j-\eta_{j-1})\rbrace, \]
for $j=2,\ldots,J-1$. Therefore, we consider the reparametrization 
$ \boldsymbol{\eta}^{\prime} = A \boldsymbol{\eta}$
between the two predictor spaces $\mathcal{S}=\lbrace \boldsymbol{\eta} \in \mathbb{R}^{J-1} | 0 \leq \eta_1<\ldots<\eta_{J-1} \rbrace$ and $\mathcal{S}^{\prime}= \mathbb{R}^{*J-1}_{+}$. Hence $ \boldsymbol{\eta}^{\prime} = AZ \boldsymbol{\beta}$ and $\boldsymbol{Y}|\boldsymbol{X}=\boldsymbol{x}$  follows the (sequential, exponential, $AZ$) model with the same parameter $\boldsymbol{\beta}$.
\end{proof}
\begin{coro}
The two subfamilies of models \{(cumulative, exponential, $Z$)$; \; Z \in \mathfrak{Z}$\} and \{(sequential, exponential, $Z$)$; \; Z \in \mathfrak{Z}$\} are equal.
\end{coro}
\begin{proof}
We use the double-inclusion method. First inclusion 
\begin{center}
\{(cumulative, exponential, $Z); \; Z \in \mathfrak{Z}$\} $ \subseteq $ \{(sequential, exponential, $Z); \; Z \in \mathfrak{Z}$\}
\end{center}
is directly obtained using property \ref{equal_cum_seq}. Noting that $A$ is invertible we obtain
\[ (\mbox{cumulative, exponential, } A^{-1}Z) = (\mbox{sequential, exponential, } Z), \]
for any design matrix $Z \in \mathfrak{Z}$ and thus the second inclusion. 
\end{proof}
It should be remarked that in general $A$ changes the constraints on space $ \mathcal{C}$ and thus the maximum likelihood. For example, the design matrix $Z_p$ corresponds to the constrained space $\mathcal{C} = \lbrace \boldsymbol{\beta}\in \mathbb{R}^{(J-1)(1+p)} | \boldsymbol{\delta_1}=\ldots=\boldsymbol{\delta_{J-1}} \rbrace $, whereas the design matrix
\[ AZ_p = \begin{pmatrix}
1 & & & & \boldsymbol{x}^t  \\
-1 & 1 & & & 0  \\
 & \ddots & \ddots  & & \vdots  \\
 & & -1 & 1 & 0
\end{pmatrix}
\]
corresponds to the constrained space $\mathcal{C} = \lbrace \boldsymbol{\beta}\in \mathbb{R}^{(J-1)(1+p)} | \boldsymbol{\delta_2}=\ldots=\boldsymbol{\delta_{J-1}}=\boldsymbol{0} \rbrace $. The design matrices $Z_p$ and $AZ_p$ are not equivalent and thus the (cumulative, exponential, $Z_p$) and (sequential, exponential, $Z_p$) models are not equivalent, whereas the (cumulative, exponential, $Z_p$) and (sequential, exponential, $AZ_p$) models are equal. In the same way, the (cumulative, exponential, $A^{-1}Z_p$) and (sequential, exponential, $Z_p$) models are equal with
\[A^{-1} Z_p = \begin{pmatrix}
1 & & & & \boldsymbol{x}^t \\
1 & 1 & & & 2\boldsymbol{x}^t  \\
\vdots & & \ddots & & \vdots  \\
1 &1 & \ldots& 1 & (J-1)\boldsymbol{x}^t
\end{pmatrix}.
\]

For the particular complete design however, there is no constraint on $\mathcal{C}$, thus $A$ cannot change it. We must simply check that $A$ does not reduce the dimension of $\mathcal{C}$. Since $A$ has full rank, the matrices $Z_c$ and $AZ_c$ are equivalent. Therefore the equality between (cumulative, exponential, $Z_c$) and (sequential, exponential, $AZ_c$) models becomes an equivalence between $(\mbox{cumulative, exponential, }Z_c)$ and $(\mbox{sequential, exponential, }Z_c)$. This equivalence described by \citet{tutz1991sequential} can thus be seen as a corollary of property \ref{equal_cum_seq}.

		\subsection{Comparison between reference and adjacent models}
		Here we follow the same approach as in the previous subsection but providing fewer details. We start from the equivalence 
\begin{equation}
\mbox{(reference, logistic, complete)} \Leftrightarrow \mbox{(adjacent, logistic, complete)}, \label{eq_agresti}
\end{equation}
shown by \citet{agresti_ordinal} and we then extend this equivalence, as in Property \ref{equal_cum_seq}, for any design matrix.
\begin{prop}
We have the equality of models
\[ (\mbox{reference, logistic}, Z) = (\mbox{adjacent, logistic}, A^tZ), \]
for any design matrix $Z \in \mathfrak{Z}$.
\end{prop}
The matrix $A^t$ turns out to be the transpose of the previously defined matrix $A$. It should be noted that canonical link function corresponds exactly to logistic cdf and reference ratio. Therefore, the subfamily of canonical models is \{(reference, logistic, $Z); \; Z \in \mathfrak{Z}\}$.
\begin{coro}\label{ref=adj}
The subfamily of canonical models \{(reference, logistic, $Z); \; Z \in \mathfrak{Z}\}$ is equal to the subfamily of models \{(adjacent, logistic, $Z); \; Z \in \mathfrak{Z}\}$.
\end{coro}
As previously, $A^t$ generally changes the constraints on space $\mathcal{C}$. For example, the design matrices $Z_p$ and $A^tZ_p$ are not equivalent. The particular equality
\[ (\mbox{reference, logistic, } (A^{t})^{-1}Z_p) = (\mbox{adjacent, logistic, } Z_p), \]
where
\[(A^{t})^{-1} Z_p = (A^{-1})^{t} Z_p = \begin{pmatrix}
1 & 1& \ldots &1 & (J-1)\boldsymbol{x}^t \\
 & \ddots& &\vdots & \vdots \\
 & & 1 & 1& 2\boldsymbol{x}^t  \\
 & & & 1 & \boldsymbol{x}^t
\end{pmatrix},
\]
corresponds to a particular reparametrization described by \citet{agresti_ordinal}. Noting that the design matrices $Z_c$ and $A^tZ_c$ are equivalent because $A^t$ has full rank, we recover the equivalence \eqref{eq_agresti} of \citet{agresti_ordinal}.

Equivalences \eqref{eq_agresti} of \citet{agresti_ordinal} and \eqref{eq_tutz} of \citep{tutz1991sequential} can be extended to equalities of subfamilies because the reparametrization is linear with respect to $\boldsymbol{\eta}$, whereas equivalences \eqref{eq_laara} of \citet{laara} and \eqref{eq_Pareto} cannot (the reparametrizations \eqref{reparametrization_laara} and \eqref{reparametrization_Pareto} are not linear with respect to $\boldsymbol{\eta}$). Finally, the four families of reference, adjacent, cumulative and sequential models can be viewed as non-disjoint sets; see figure \ref{patates}, where $\mathcal{L}$ (respectively $\mathcal{E}$, $\mathcal{G}_{\mbox{-}}$ and $\mathcal{P}_a$) denote the logistic (respectively exponential, Gumbel min and Pareto) cdf.

\begin{figure}[h!]
\centering
\caption{Families intersections and models equivalences \eqref{eq_laara} of \citet{laara} (stars), \eqref{eq_tutz} of \citet{tutz1991sequential} (dots), \eqref{eq_agresti} of \citet{agresti_ordinal} (squares) and \eqref{eq_Pareto} (crosses).}\label{patates}
\includegraphics[width=1\textwidth]{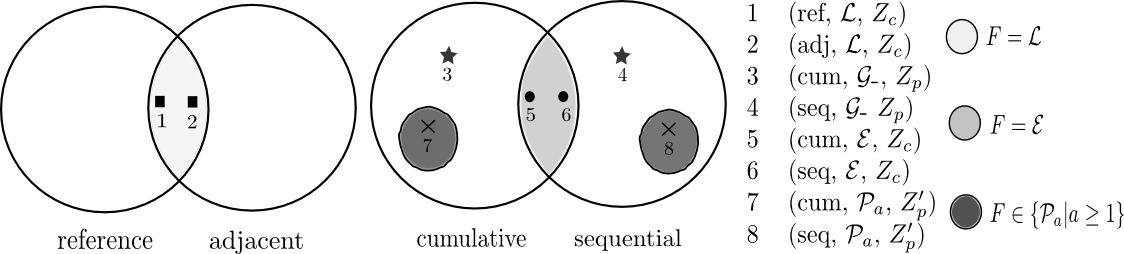}
\end{figure}


	\section{\label{invariance} Permutation invariance and stability}
	Corollary \ref{ref=adj} shows an equality between two families, one dedicated to nominal data and the other one to ordinal data, since the adjacent ratio uses the category ordering assumption whereas the reference ratio does not. These two families of reference and adjacent models overlap for the particular case of the logistic cdf (see figure \ref{patates}). Thus, we can wonder whether this subfamily of canonical models is more appropriate for nominal or ordinal data. More generally, we want to classify each $(r,F,Z)$ triplet as a nominal or an ordinal model. \citet{mccullagh78} proposed that models for nominal categories should be invariant under arbitrary permutations and that models for ordinal data should be invariant only under the special reverse permutation. We therefore investigate permutation invariances for each family of models.

\vspace*{0.2cm}

Let us first introduce the notion of permutation invariance. Each index $j \in \lbrace 1,\ldots,J \rbrace$ is associated with a particular category. Modifying this association potentially changes the model. Such a modification is characterized by a permutation $\sigma$ of $\lbrace 1,\ldots,J \rbrace$. The permuted vector parameter $\boldsymbol{\pi}_{\sigma}=(\pi_{\sigma(1)}, \ldots, \pi_{\sigma(J-1)}) $ is thus introduced and the permuted model is summarized by
\[ r(\boldsymbol{\pi}_{\sigma})=\mathcal{F}(Z \boldsymbol{\beta}), \]
and is denoted by the indexed triplet $(r,F,Z)_{\sigma}$. 
\begin{Def}
Let $\sigma$ be a permutation of $\lbrace 1,\ldots,J \rbrace$. A model $(r,F,Z)$ is said to be invariant under $\sigma$ if 
the models $(r,F,Z)$ and $(r,F,Z)_{\sigma}$ are equivalent. A family of models $\mathfrak{M}$ is said to be stable under $\sigma$ if $\sigma(\mathfrak{M}) = \mathfrak{M}$, where $ \sigma(\mathfrak{M}) = \{ (r,F,Z)_{\sigma} |\; (r,F,Z) \in \mathfrak{M} \}$.
\end{Def}

		\subsection{Invariances of reference models}
		For the reference ratio the probability of each category is only connected with the probability of reference category $J$. Thus, we focus on permutations that fix the reference category.
\begin{prop}\label{prop_ref}
Let $\sigma_J$ be a permutation of $\{1,\ldots,J\}$ such that $\sigma_J(J)=J$ and let $P_{\sigma_J}$ be the restricted permutation matrix of dimension $J-1$
\[ (P_{\sigma_J})_{i,j} = \left\{
\begin{array}{ll}
1 & \mbox{if } \; i=\sigma_J(j), \\
0 & \mbox{otherwise},
\end{array}
\right.
\]
for $i, j \in \{1,\ldots,J-1\}$. Then we have
\[ (\mbox{reference}, F, Z)_{\sigma_J} = (\mbox{reference}, F, P_{\sigma_J}Z), \]
for any $F \in \mathfrak{F}$ and any $Z \in \mathfrak{Z}$.
\end{prop}
\begin{proof}
For the reference ratio we have
\[ r_j(\boldsymbol{\pi}_{\sigma_J}) = r_{\sigma_J(j)}(\boldsymbol{\pi}), \]
for $j \in \{1,\ldots,J-1\}$. Thus we simply need to permute the linear predictors using $P_{\sigma_J}$ and we obtain $\boldsymbol{\eta}^{\prime} = P_{\sigma_J}\boldsymbol{\eta} =P_{\sigma_J}Z \boldsymbol{\beta} $.
\end{proof}
\noindent Noting that $P_{\sigma_J}$ is invertible with $P_{\sigma_J}^{-1}=P_{\sigma_J^{-1}}$, we get:
\begin{coro}
The family of reference models is stable under the $(J-1)!$ permutations that fix the reference category.
\end{coro}

\begin{coro}\label{ref_invariance}
Let $F \in \mathfrak{F}$. The particular (reference, $F$, complete) and (reference, $F$, proportional) models are invariant under the $(J-1)!$ permutations that fix the reference category.
\end{coro}
\begin{proof}
The design matrices $Z_c$ and $P_{\sigma_J}Z_c$ are equivalent because $P_{\sigma_J}$ has full rank. Therefore, (reference, $F$, $Z_c)_{\sigma_J}$ and (reference, $F$, $Z_c$) models are equivalent, which means that (reference, $F$, $Z_c$) model is invariant under $\sigma_J$. Moreover, the design matrices $Z_p$ and $P_{\sigma_J}Z_p$ are also equivalent. In fact, permutation $\sigma_J$ does not change the contrast of common slope
\[ \boldsymbol{\delta_1}=\ldots=\boldsymbol{\delta_{J-1}} \Leftrightarrow \boldsymbol{\delta_{\sigma_J(1)}}=\ldots=\boldsymbol{\delta_{\sigma_J(J-1)}}. \]
\end{proof}
\noindent According to property \ref{prop_ref}, reference link functions are invariant under the $(J-1)!$ permutations that fix the reference category. But now what happens if we transpose the reference category? Canonical link functions have specific behaviour in this case.

\subsubsection*{Invariances of canonical models} The logistic cdf gives a particular property to reference models.
\begin{prop}\label{inv_canonic}
Let $\tau$ be a non identical transposition of the reference category $J$ and $B_{\tau}$ be the $(J-1)$-identity matrix, whose $\tau(J)^{\mbox{{\footnotesize th}}}$ column is replaced by a column of $-1$. Then we have
\[ (\mbox{reference, logistic, }Z)_{\tau} = (\mbox{reference, logistic, }B_{\tau}Z), \]
for any design matrix $Z \in \mathfrak{Z}$.
\end{prop}
\begin{proof}
Assume that the distribution of $\boldsymbol{Y}|\boldsymbol{X}=\boldsymbol{x}$ is defined by the transposed canonical (reference, logistic, $Z)_{\tau}$ model. Thus we obtain
\[ \left\{
\begin{array}{ll}
\displaystyle \frac{\pi_j}{\pi_{\tau(J)}} = \exp(\eta_j) & \mbox{for}\; j \neq J \; \mbox{and} \; j \neq \tau(J), \\
\displaystyle \frac{\pi_J}{\pi_{\tau(J)}} = \exp(\eta_{\tau(J)}), &
\end{array}
\right.
\]
or equivalently
\[ \left\{
\begin{array}{ll}
\displaystyle \frac{\pi_j}{\pi_{J}} = \frac{\pi_j}{\pi_{\tau(J)}} \frac{\pi_{\tau(J)}}{\pi_{J}}= \exp(\eta_j-\eta_{\tau(J)}) & \mbox{for} j \neq J \; \mbox{and} \; j \neq \tau(J), \\
\displaystyle \frac{\pi_{\tau(J)}}{\pi_J} = \exp(-\eta_{\tau(J)}). &
\end{array}
\right.
\]
Hence $\boldsymbol{Y}|\boldsymbol{X}=\boldsymbol{x}$ follows the canonical (reference, logistic, $B_{\tau}Z$) model.
\end{proof}
\noindent Noting that $B_{\tau}$ is invertible with $B_{\tau}^{-1}=B_{\tau}$ we get:
\begin{coro}\label{canonic}
The subfamily of canonical models is stable under all permutations.
\end{coro}
\noindent Noting that matrices $P_{\sigma_J}$ and $B_{\tau}$ have full rank, we get:
\begin{coro}\label{canonical_model}
The canonical model (reference, logistic, complete) is invariant under all permutations.
\end{coro}

			\subsubsection*{Non-invariance}
			
\begin{conj}
Let $F \in \mathfrak{F} \setminus \mathcal{L}$ (respectively $F \in \mathfrak{F}$), where $\mathcal{L}$ denotes the logistic cdf. The particular (reference, $F$, complete) (respectively (reference, $F$, proportional)) model is invariant only under the $(J-1)!$ permutations that fix the reference category.
\end{conj}
\noindent The (reference, $F$, complete) and (reference, $F$, proportional) models are invariant under the $(J-1)!$ permutations that fix the reference category (Corollary \ref{ref_invariance}). But are they still invariant under other permutations? Non-invariance of models may be shown when $F$ is analytically defined. This is more complex for Cauchy distribution for instance. Figure \ref{ref_perm} therefore investigates non-invariance of reference models on a benchmark dataset. We use, here and in the following, the boy’s disturbed dreams benchmark dataset drawn from a study that cross-classified boys by their age $x$ and the severity of their disturbed dreams $y$ ($4$ ordered categories) ; see \citet{maxwell1961analysing}. The canonical (reference, logistic, complete) model, which is invariant under all permutations, is a special case (represented in blue in figure \ref{ref_perm}.a). For other cases we obtain $J!/(J-1)!=J=4$ plateaus as expected. Each plateau corresponds to a particular reference category with the $(J-1)!=6$ permutations that fix this category.

	\begin{figure}[!h]
	\centering
	\caption{Ordered log-likelihood of a. (reference, $F$, complete$)_{\sigma}$ and b. (reference, $F$, proportional$)_{\sigma}$ models for all permutations $\sigma$ and for logistic (blue), Cauchy (green) and Gumbel min (red) cdfs.}\label{ref_perm}
    	\includegraphics[width=0.48\textwidth]{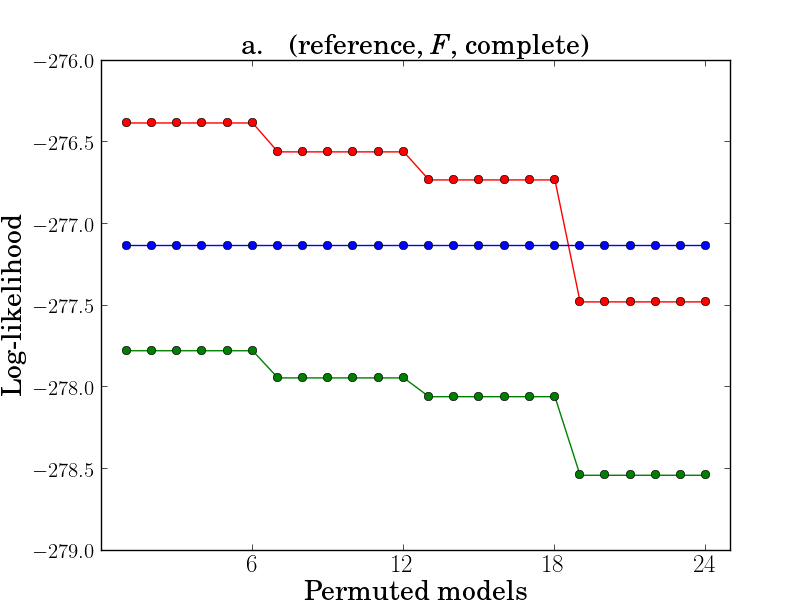}
    	\includegraphics[width=0.48\textwidth]{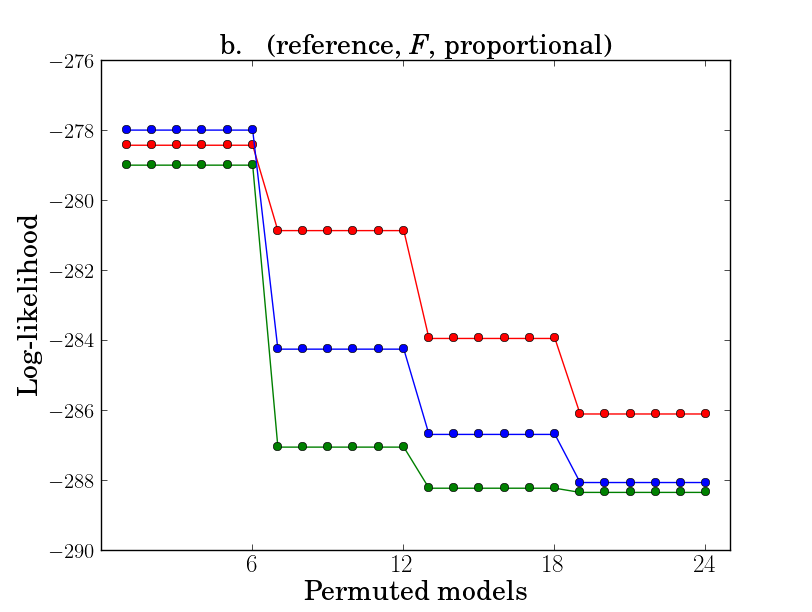}
    	\end{figure}

\noindent Given the reference category, ordering assumption has no impact on reference model fit as expected. We have shown that reference models are more or less invariant under transpositions of the reference category. Reference models can thus be classified according to their degree of invariance; see figure \ref{nominal_scale}. 

\begin{figure}[!h]
\centering
\caption{Classification of reference models on an invariance scale.}\label{nominal_scale}
\includegraphics[width=1\textwidth]{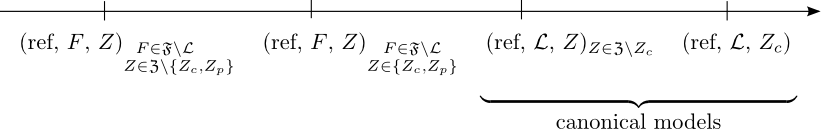}
\end{figure}

		\subsection{Invariances of adjacent and cumulative models}
		Adjacent, cumulative and sequential ratios are defined using the category ordering assumption and are thus naturally devoted to ordinal data. According to \citet{mccullagh78}, models for ordinal data should be invariant only under the special reverse permutation. In particular, \citet{mccullagh80} showed that the three models (cumulative, logistic, proportional), (cumulative, normal, proportional) and (cumulative, Cauchy, proportional) are invariant under the reverse permutation. We generalize this property to any symmetrical cdf $F$ and also for the complete design. The adjacent ratio turns out to have the same property, but not the sequential ratio.
\begin{prop}
Let $\tilde{\sigma}$ be the reverse permutation (i.e. $\tilde{\sigma}(j)=J-j+1,\; \forall j \in \{1,\ldots,J\}$) and $\tilde{P}$ be the restricted reverse permutation matrix of dimension $J-1$
\[ \tilde{P} = \begin{pmatrix}
 & & 1\\
 &  \revdots &  \\
1 & &
\end{pmatrix}.
\]
Then we have
\[ \mbox{(adjacent, }F, Z)_{\tilde{\sigma}} = \mbox{(adjacent, }\tilde{F}, -\tilde{P}Z), \]
and
\[ \mbox{(cumulative, }F, Z)_{\tilde{\sigma}} = \mbox{(cumulative, }\tilde{F}, -\tilde{P}Z), \]
for any cdf $F \in \mathfrak{F}$ and any design matrix $Z \in \mathfrak{Z}$.
\end{prop}
\begin{proof}
For adjacent and cumulative ratios, it can be shown that
\begin{equation}
 r_j(\boldsymbol{\pi}_{\tilde{\sigma}}) = 1- r_{\tilde{\sigma}(j+1)}(\boldsymbol{\pi}), \label{complement}
\end{equation}
for $j \in \{1,\ldots,J-1\}$. Assume that the distribution $\boldsymbol{Y}|\boldsymbol{X}=\boldsymbol{x}$ is defined by the permuted $(r,F,Z)_{\tilde{\sigma}}$ model with $r=$ adjacent or cumulative, i.e.
\[ r_j(\boldsymbol{\pi}_{\tilde{\sigma}}) = F(\eta_j),\] 
for $j \in \{1,\ldots,J-1\}$. Using equality \eqref{complement}, we obtain equivalently
\begin{equation}
 r_{J-j}(\boldsymbol{\pi}) = \tilde{F}(-\eta_j), \label{symm} 
\end{equation}
for $j \in \{1,\ldots,J-1\}$. Denoting $i=J-j$, \eqref{symm} becomes
\begin{equation}
 r_i(\boldsymbol{\pi}) = \tilde{F}(-\eta_{J-i}), \nonumber
\end{equation}
for $i \in \{1,\ldots,J-1\}$. Hence $\boldsymbol{Y}|\boldsymbol{X}=\boldsymbol{x}$ follows the $(r,\tilde{F},-\tilde{P}Z)$ model.
\end{proof}
\noindent Noting that $\tilde{P}$ is invertible with $\tilde{P}^{-1}=\tilde{P}$, we get:
\begin{coro}\label{cum_adj_stability}
The two families of adjacent and cumulative models are stable under the reverse permutation.
\end{coro}
\begin{coro}\label{cum_adj_Zc_Zp}
Let $F \in \tilde{\mathfrak{F}}$. The particular (adjacent, $F$, complete), (adjacent, $F$, proportional), (cumulative, $F$, complete) and (cumulative, $F$, proportional) models are invariant under the reverse permutation.
\end{coro}
\begin{proof}
The design matrices $Z_p$ and $-\tilde{P}Z_p$ are equivalent because
\[ \boldsymbol{\delta_1}=\ldots=\boldsymbol{\delta_{J-1}} \Leftrightarrow -\boldsymbol{\delta_{J-1}}=\ldots=-\boldsymbol{\delta_{1}}. \]
The design matrices $Z_c$ and $-\tilde{P}Z_c$ are also equivalent because $\tilde{P}$ has full rank.
\end{proof}

			\subsubsection{Non-invariance}

\begin{conj}
Let $F \in \tilde{\mathfrak{F}}$ (respectively $F \in \tilde{\mathfrak{F}} \setminus \mathcal{L}$). The particular (adjacent, $F$, proportional), (cumulative, $F$, complete) and (cumulative, $F$, proportional) (respectively (adjacent, $F$, complete)) models are invariant only under the reverse permutation.
\end{conj}
\noindent We want to show that adjacent and cumulative models are not invariant under other permutations than the reverse permutation. Figure \ref{adj_perm} (respectively \ref{cum_perm}) investigates the case of (adjacent, $F$, complete) and (adjacent, $F$, proportional) models (respectively (cumulative, $F$, complete) and (cumulative, $F$, proportional) models) for all the $J!=24$ permutations and for three different cdfs. We obtain $J!/2!=12$ plateaus when $F$ is symmetric, as expected. Each plateau corresponds to a particular permutation and its associated reverse permutation. The particular (adjacent, logistic, complete) model is equivalent to the canonical (reference, logistic, complete) model (see equivalence of Agresti \eqref{eq_agresti}) and thus is invariant under all permutations according to corollary \ref{canonical_model} (see the blue line in figure \ref{adj_perm}a.). Let us note that log-likelihood has diverged for some cumulative models because of non-positive probabilities (see section \ref{compatibility} for more details) and thus are not represented in figure \ref{cum_perm}.

	\begin{figure}[!h]
	\centering
	\caption{Ordered log-likelihood of a. (adjacent, $F$, complete$)_{\sigma}$ and b. (adjacent, $F$, proportional$)_{\sigma}$ models for all permutations $\sigma$ and for logistic (blue), Cauchy (green) and Gumbel min (red) cdfs.}\label{adj_perm}
    	\includegraphics[width=0.48\textwidth]{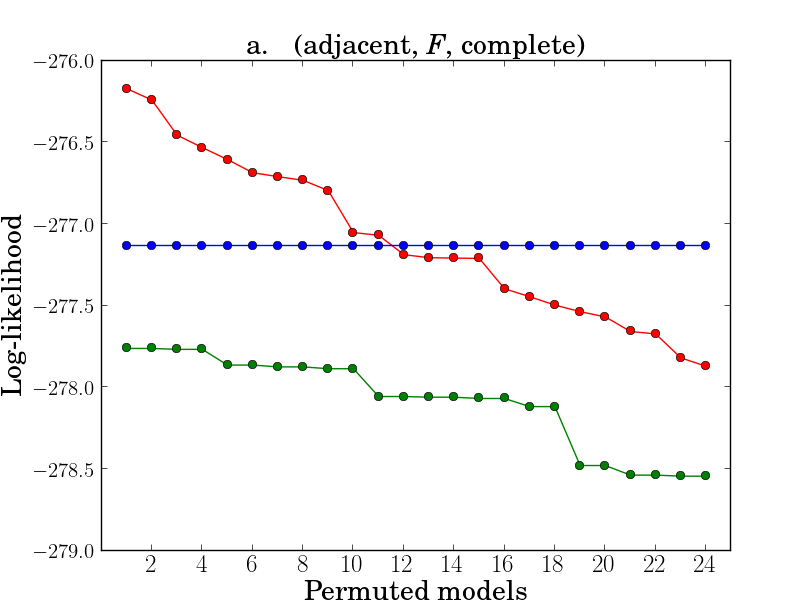}
    	\includegraphics[width=0.48\textwidth]{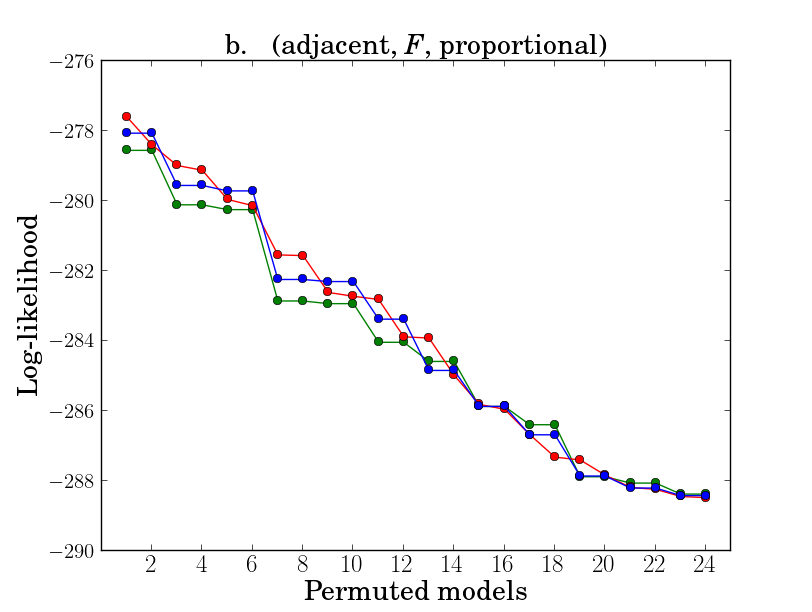}
    	\end{figure}

   	\begin{figure}[!h]
   	\centering
    	\caption{Ordered log-likelihood of a. (cumulative, $F$, complete$)_{\sigma}$ and b. (cumulative, $F$, proportional$)_{\sigma}$ models for all permutations $\sigma$ and for logistic (blue), Cauchy (green) and Gumbel min (red) cdfs.}\label{cum_perm}
    	\includegraphics[width=0.48\textwidth]{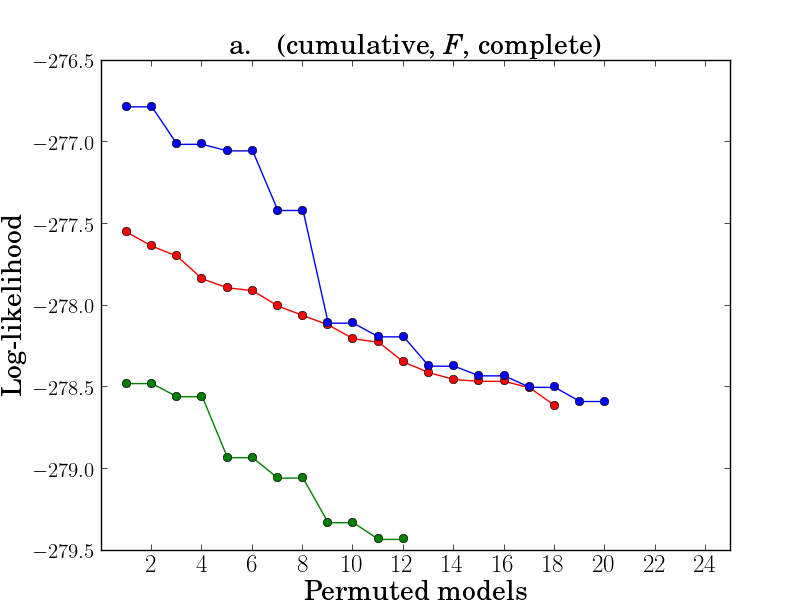}
    	\includegraphics[width=0.48\textwidth]{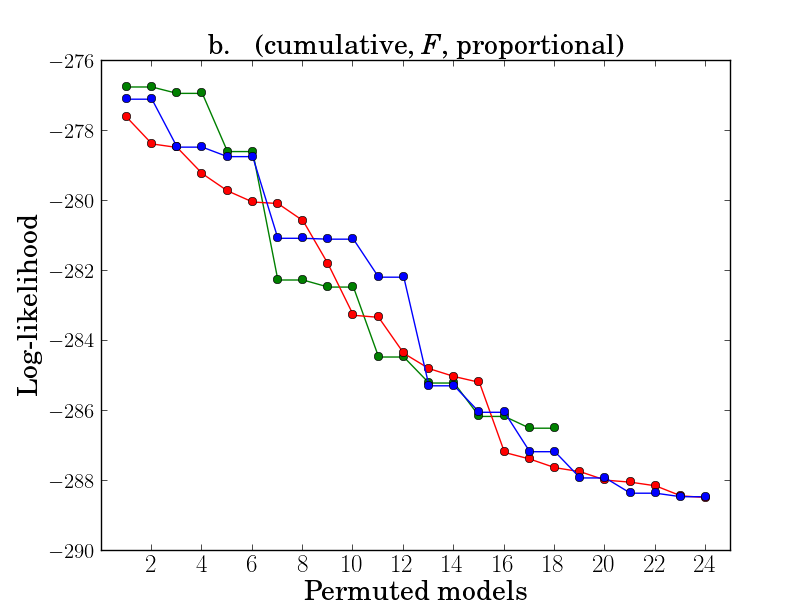}
    	\end{figure} 

			\subsubsection{Reversible models}
			A stochastic process is defined as time-reversible if it is invariant under the reversal of the time scale. Assume that in our context the order between categories plays the role of the time scale. By analogy with stochastic processes, we state the following definition:
\begin{Def}
A model for ordinal data is reversible if $r_j(\boldsymbol{\pi}_{\tilde{\sigma}}) = 1- r_{\tilde{\sigma}(j+1)}(\boldsymbol{\pi})$ for the reverse permutation $\tilde{\sigma}$ (equation \eqref{complement}).
\end{Def}

\noindent This entails a reversal of the cdf $F$ (e.g. the reversal of a Gumbel min cdf is a Gumbel max cdf) and design matrix $Z$. In the case of symmetric cdf $F$ and complete or proportional design matrix, a reversible models is invariant under the reverse permutation (see corollary \ref{cum_adj_Zc_Zp}). Let us note that (adjacent, logistic, complete) model, being invariant under all permutations, is inappropriate for ordinal data. More generally, the subfamily of models \{(adjacent, logistic, $Z); \; Z \in \mathfrak{Z}$\}, being stable under all permutations (corollaries \ref{ref=adj} and \ref{canonic}), is inappropriate for ordinal data.

Adjacent and cumulative models are thus reversible (see corollary \ref{cum_adj_stability}) while sequential models are not. This distinction between these two categories of models for ordinal data is critical for a clear understanding of the invariance properties.

		\subsection{Invariances of sequential models}
The reversibiliy property distinguishes sequential models from adjacent and cumulative models because equality (\ref{complement}) is no longer valid for sequential models. The reverse permutation changes the structure of a sequential model. Nevertheless we have the following property:
\begin{prop}
Let $\tilde{\tau}$ be the transposition of the last two categories and $A_{\tilde{\tau}}$ be the square matrix of dimension $J-1$
\[ A_{\tilde{\tau}} = \begin{pmatrix}
1 & & & \\
 & \ddots & & \\
 & & 1 & \\
 & &   & -1
\end{pmatrix}.
\]
Then we have
\[ (\mbox{sequential, }F, Z)_{\tilde{\tau}} = (\mbox{sequential, }F, A_{\tilde{\tau}}Z),\]
for any $F \in \tilde{\mathfrak{F}}$ and any design matrix $Z \in \mathfrak{Z}$.
\end{prop}
\begin{proof}
Assume that the distribution of $\boldsymbol{Y}|\boldsymbol{X}=\boldsymbol{x}$ is defined by the transposed (sequential, $F$, $Z)_{\tilde{\tau}}$ model with a symmetric cdf $F$, i.e.
\[ \frac{\pi_{\tilde{\tau}(j)}}{\pi_{\tilde{\tau}(j)}+ \ldots + \pi_{\tilde{\tau}(J-1)} + \pi_{\tilde{\tau}(J)} } = F(\eta_j) , \]
for $j \in \{1,\ldots,J-1\}$. The last equation can be isolated
\[ \left\{
\begin{array}{ll}
\displaystyle \frac{\pi_j}{\pi_j + \ldots + \pi_J} = F(\eta_j) & \forall j \in \{1,\ldots,J-2\}, \\
\displaystyle \frac{\pi_J}{\pi_J + \pi_{J-1}} = F(\eta_{J-1}).
\end{array}
\right.
\]
This is equivalent to
\[ \left\{
\begin{array}{ll}
\displaystyle \frac{\pi_j}{\pi_j + \ldots + \pi_J} = F(\eta_j), & \forall j \in \{1,\ldots,J-2\}, \\
\displaystyle \frac{\pi_{J-1}}{\pi_{J-1} + \pi_J} = \tilde{F}(-\eta_{J-1}).
\end{array}
\right.
\]
Since $F$ is symmetric (i.e. $F=\tilde{F}$), then $\boldsymbol{Y}|\boldsymbol{X}=\boldsymbol{x}$ follows the (sequential, $F$, $A_{\tilde{\tau}}Z$) model.
\end{proof}
\noindent Noting that $A_{\tilde{\tau}}$ is invertible with $A_{\tilde{\tau}}^{-1}=A_{\tilde{\tau}}$, we get:
\begin{coro}
The subfamily of models \{(sequential, $F$, $Z$)$; \; F \in \tilde{\mathfrak{F}}$, $Z \in \mathfrak{Z}$\} is stable under the transposition $\tilde{\tau}$ of the last two categories.
\end{coro}
\noindent Noting that matrix $A_{\tilde{\tau}}$ has full rank, we get:
\begin{coro}\label{seq_tau_comp}
Let $F \in \tilde{\mathfrak{F}}$. The particular (sequential, $F$, complete) model is invariant under the transposition $\tilde{\tau}$ of the last two categories.
\end{coro}
However the design matrices $A_{\tilde{\tau}}Z_p$ and $Z_p$ are not equivalent. Therefore, the (sequential, $F$, proportional) model is not invariant under $\tilde{\tau}$, even if $F$ is symmetric.

			\subsubsection{Non-invariance}
\begin{conj}
Let $F \in \tilde{\mathfrak{F}}$. The particular (sequential, $F$, complete) model is invariant only under the transposition $\tilde{\tau}$ of the last two categories.
\end{conj}
The (sequential, $F$, complete) models are invariant under the transposition of the last two categories when $F$ is symmetric (Corollary \ref{seq_tau_comp}). But are they still invariant under other permutations? Figure \ref{seq_perm} investigates the case of (sequential, $F$, complete) models for all the $J!=24$ permutations and for different cdfs. We obtain $J!/2!=12$ plateaus only for symmetrical cdfs, as expected. Each plateau corresponds to a particular permutation and its associated transposition of the last two categories.

	\begin{figure}[h!]
	\centering
	\caption{Ordered log-likelihood of (sequential, $F$, complete$)_{\sigma}$ models for all permutations $\sigma$ and for logistic (blue), Cauchy (green) and Gumbel min (red) cdfs.}\label{seq_perm}
    	\includegraphics[width=0.48\textwidth]{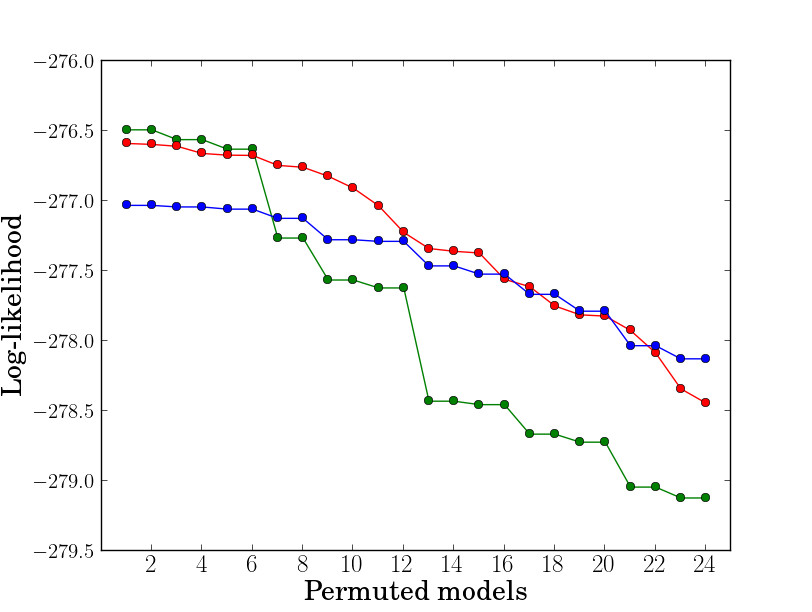}
    	\end{figure}

\vspace*{0.5cm}

In general, sequential models are not invariant under a permutation of categories except the particular case of symmetric cdf and complete design matrix, where sequential models are invariant under the transposition of the last two categories.

\section{\label{applications} Applications}

	\subsection{Supervised classification}
	Linear, quadratic and logistic discriminant analyses are three classical methods used for supervised classification. The logistic discriminant analysis often outperforms other classical methods \citep{classification2000}. In our context, we prefer to consider the logistic regression rather than the logistic discriminant analysis, these two methods being very similar \citep{hastie2005}. In this subsection, we propose a family of classifiers that includes the logistic regression as a particular case. We then compare the classification error rates on three benchmark datasets (available on UCI), using $10$-fold cross validation. The logistic regression is fully specified by the (reference, logistic, complete) triplet. 
	
	We propose to use the entire set of reference models with a complete design, which all have the same number of parameters. We can change the cdfs $F$ to obtain a better fit. For the application, we used ten different cdf $F$:
\[ \mathfrak{F}_0 = \{ \mbox{normal, Laplace, Gumbel min, Gumbel max, Student}(1), \ldots, \mbox{Student}(6) \}, \]
from which ten classifiers were built
\[ \mathfrak{C}^* = \{ (\mbox{reference}, F, \mbox{complete}); \; F \in \mathfrak{F}_0 \}. \]
All these classifiers were compared with the logistic regression, using $10$-fold cross validation. For each classifier, the mean error rate was computed on the ten sub-samples and compared with the logistic regression error rate (represented in blue in figures \ref{thyroid}, \ref{vehicle} and \ref{pages_blocks}). The impact of changing the cdf can be seen (the corresponding minimal error rate is represented in green).

	In section \ref{invariance}, we saw that (reference, $F$, complete) models, with $F \neq$ logistic, do not seem to be invariant under transpositions of the reference category. This means that changing the reference category potentially changes the fit. Therefore, we propose to extend the set of classifiers $\mathfrak{C}^*$ to obtain
\[ \mathfrak{C} = \{ (\mbox{reference}, F, \mbox{complete})_{\tau} ; \; F \in \mathfrak{F}_0, \; \tau \in \mathcal{T}_J \}, \]
where $\mathcal{T}_J$ contains all transpositions of the reference category $J$. Finally, the set $ \mathfrak{C}$ contains exactly $ 10 \times J$ classifiers. All these classifiers were then compared with the logistic regression. The impact of changing the reference category can be seen (the corresponding minimal error rate is represented in red). The three following original datasets were extracted from the UCI machine learning repository and the datasets already partitioned by means of a 10-fold cross validation procedure were extracted from the KEEL dataset repository. They contain respectively $3$, $4$ and $5$ classes.

		\paragraph*{Thyroid}
		This dataset is one of the several thyroid databases available in the UCI repository. The task is to detect if a given patient is normal (1) or suffers from hyperthyroidism (2) or hypothyroidism (3). This dataset contains $n=7200$ individuals and all 21 explanatory variables are quantitative.
		
		For the (reference, logistic, complete) model, the mean error rate of misclassification (in blue) was $6.12\%$. Using all the classifiers of $\mathfrak{C}^*$, the best classifier was the (reference, Student(3), complete) model with a misclassification mean error rate (in green) of $2.32\%$ (see figure \ref{thyroid} a.). Finally, using all the classifiers of $\mathfrak{C}$, the best classifier was the (reference, Student(2), complete$)_{\tau}$ model (where $\tau(J)=2$) with a misclassification mean error rate (in red) of $1.61\%$. The gain appeared to be mainly due to the change in cdf $F$ (see figure \ref{thyroid} b.).

	\begin{figure}[!h]
	\centering
	\caption{\label{thyroid} Error rates for the classifiers of a. $\mathfrak{C}^*$ and b. $\mathfrak{C}$ on the thyroid dataset.}
    	\includegraphics[width=0.49\textwidth]{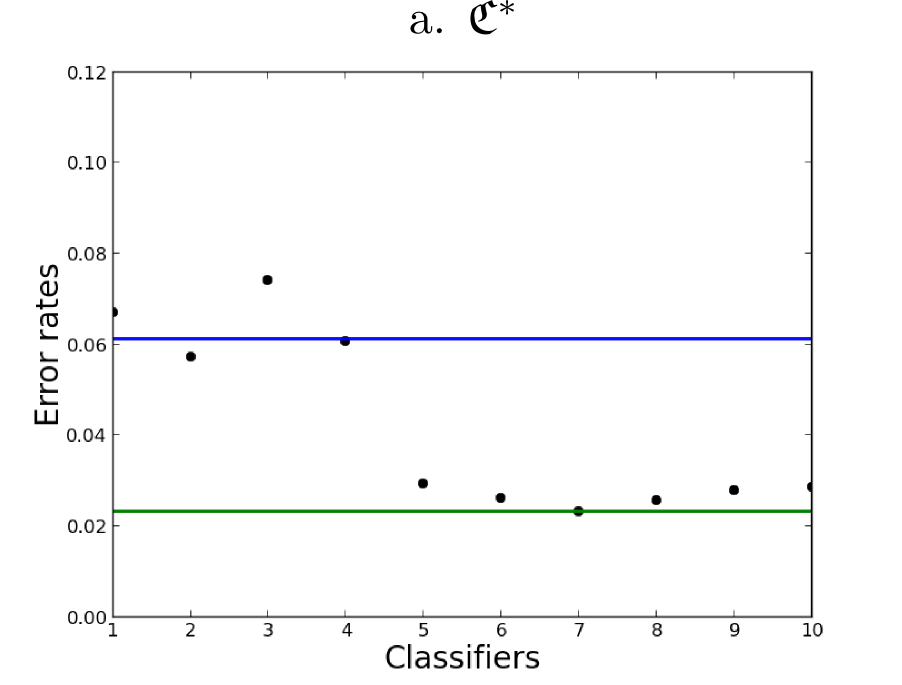}
    	\includegraphics[width=0.49\textwidth]{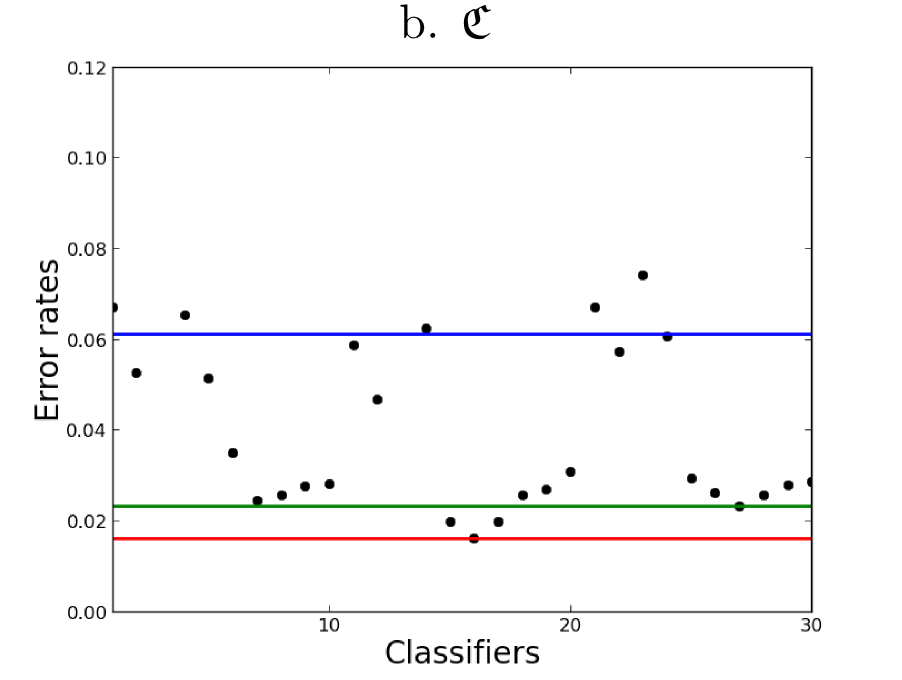}
    	\end{figure}		

		\paragraph*{Vehicle}
		The purpose is to classify a given silhouette as one of four types of vehicle, using a set of features extracted from the silhouette. The vehicle may be viewed from one of many different angles. The four types of vehicle are: bus (1), opel (2), saab (3) and van (4). This dataset contains $n=846$ instances and all $18$ explanatory variables are quantitative.
		
		For the (reference, logistic, complete) model, the misclassification mean error rate (in blue) was $19.03\%$. All classifiers of $\mathfrak{C}^*$ or $\mathfrak{C}$ gave higher error rate (see figure \ref{vehicle}).

    	\begin{figure}[!h]
	\centering
    	\caption{\label{vehicle} Error rates for the classifiers of a. $\mathfrak{C}^*$ and b. $\mathfrak{C}$ on the vehicle dataset.}
    	\includegraphics[width=0.49\textwidth]{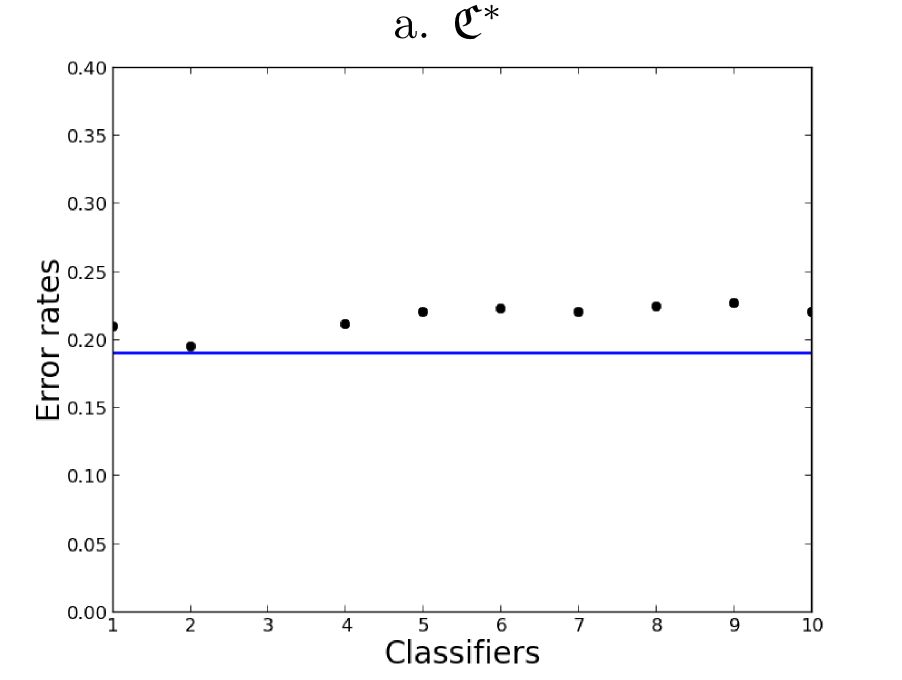}
    	\includegraphics[width=0.49\textwidth]{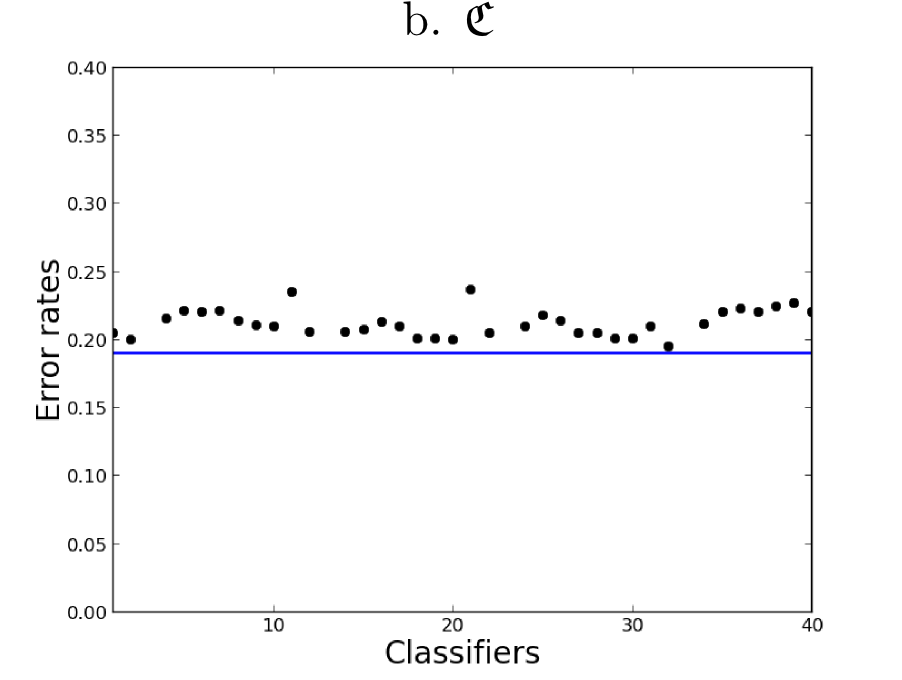}
    	\end{figure}

		\paragraph*{Pages blocks}
		The problem consists in classifying all the blocks of page layout in a document detected by a segmentation process. This is an essential step in document analysis to separate text from graphic areas. The five classes are: text (1), horizontal line (2), picture (3), vertical line (4) and graphic (5). The $n = 5473$ examples were extracted from 54 distinct documents. Each observation concerns one block. All 10 explanatory variables are quantitative.
		
		For the (reference, logistic, complete) model, the misclassification mean error rate (in blue) was $5.55\%$. Using all the classifiers of $\mathfrak{C}^*$, the best classifier was the (reference, Student(3), complete) model with a misclassification mean error rate (in green) of $3.67\%$ (see figure \ref{pages_blocks} a.). Finally, using all the classifiers of $\mathfrak{C}$, the best classifier was the (reference, Student(1), complete$)_{\tau}$ model (where $\tau(J)=1$) with a misclassification mean error rate (in red) of $2.94\%$ (see figure \ref{pages_blocks} b.).

	\begin{figure}[!h]
	\centering
    	\caption{\label{pages_blocks} Error rates for the classifiers of a. $\mathfrak{C}^*$ and b. $\mathfrak{C}$ on the pages-blocks dataset.}
    	\includegraphics[width=0.49\textwidth]{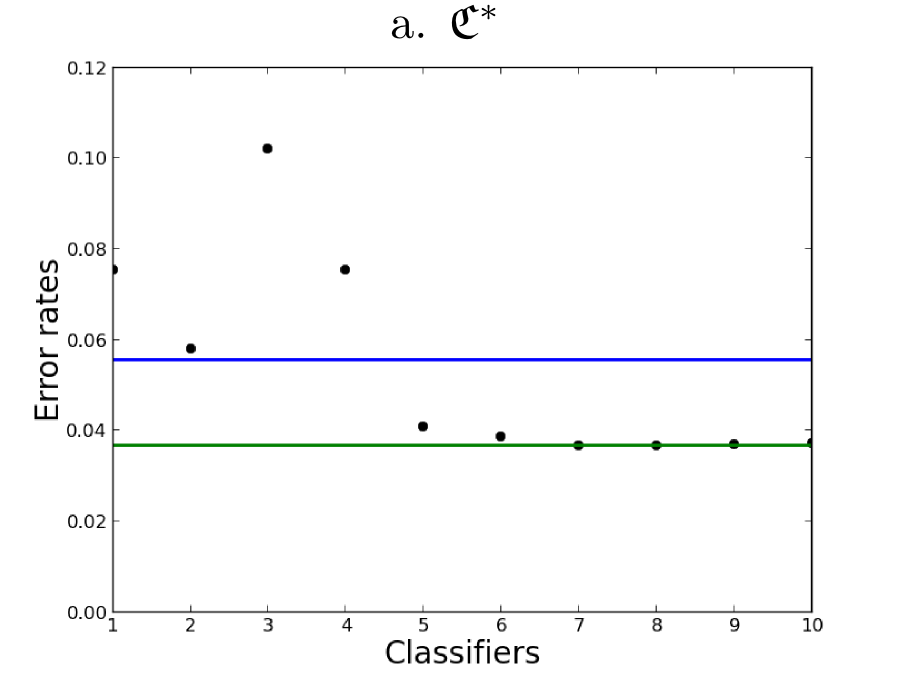}
    	\includegraphics[width=0.49\textwidth]{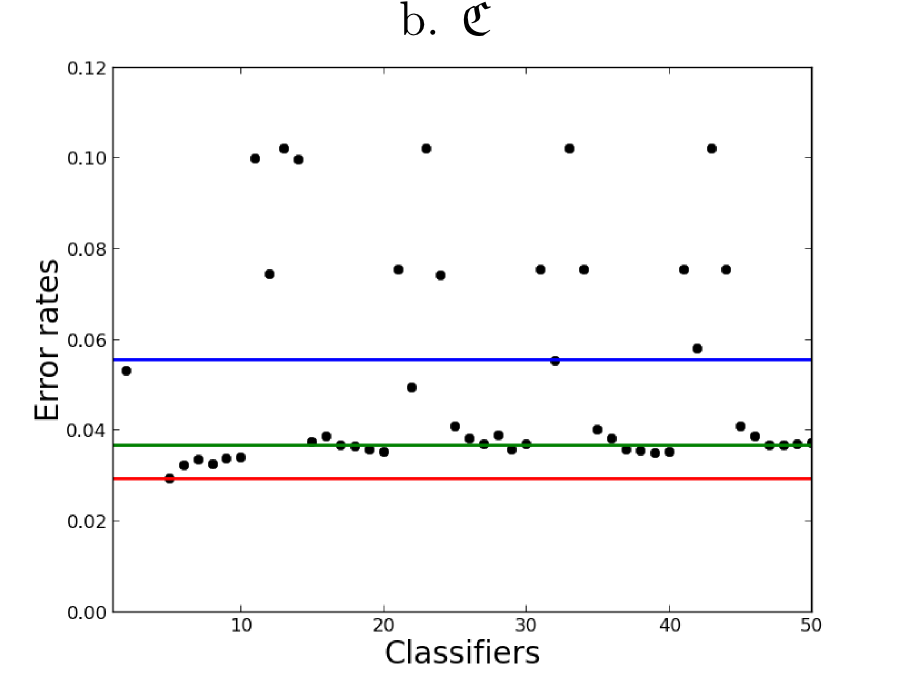}
    	\end{figure}	
    	

The Gumbel min cdf gave the worst classifiers. The Normal, Laplace and Gumbel max cdfs were comparable to the logistic cdf. Finally, Student distributions outperformed the other cdfs, likely because of their heavier tails.

	\subsection{Partially-known total ordering}
	Let us first briefly introduce the pear tree dataset. Harvested seeds of \textit{Pyrus spinosa} were sown and planted in January 2001 in a nursery located near Aix-en Provence, southeastern France. In winter 2001, the first annual shoots of the trunk of 50 one-year-old individuals were described by node. The presence of an immediate axillary shoot was noted at each successive node. Immediate shoots were classified in four categories according to length and transformation or not of the apex into spine (i.e. definite growth or not). The final dataset was thus constituted of 50 bivariate sequences of cumulative length 3285 combining a categorical variable $Y$ (type of axillary production selected from among latent bud (l), unspiny short shoot (u), unspiny long shoot (U), spiny short shoot (s) and spiny long shoot (S)) with an interval-scaled variable $X$ (internode length); see figure \ref{pear_tree_legende}.
	
	\begin{figure}[h!]
	\centering
	\caption{Axillary production of pear tree.}\label{pear_tree_legende}
    	\includegraphics[width=0.37\textwidth]{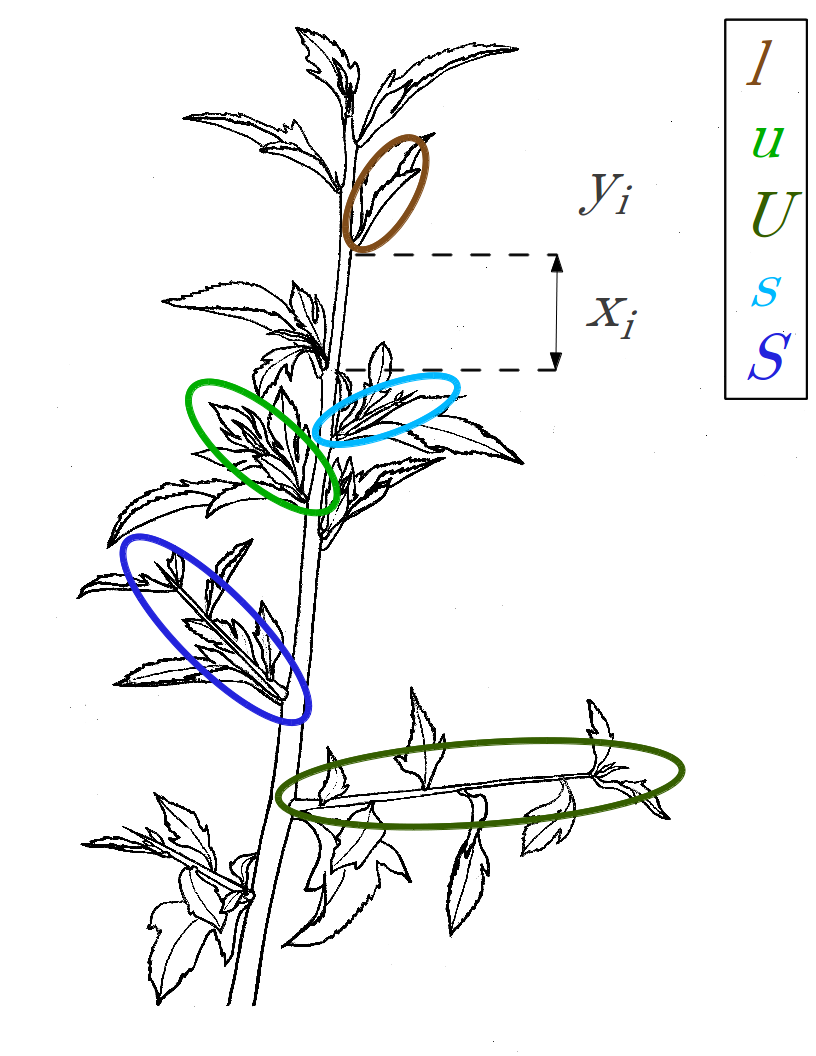}
    	\end{figure}	
	
	We sought to explain the type of axillary production as a function of the internode length, using only partial information about category ordering. In fact, the three categories $l$, $u$ and $U$ (respectively $l$, $s$ and $S$) are ordered. But the two mechanisms of elongation and transformation into spine are difficult to compare. Thus the order among the five categories was only partially known and can be summarized by the Hasse diagram in figure \ref{hasse_diagram}. However, total ordering among the five categories was assumed and we attempted to recover it. We therefore tested all permutations of the five categories such that $l < u < U$ and $l < s < S$ (i.e. only $4!/2!2!=6$ permutations). Since axillary production may be interpreted as a sequential mechanism, we use the sequential ratio.

	\begin{figure}[h!]
	\centering
	\caption{Hasse diagram of order relationships $l < u < U$ and $l < s < S$.}\label{hasse_diagram}
    	\includegraphics[width=0.2\textwidth]{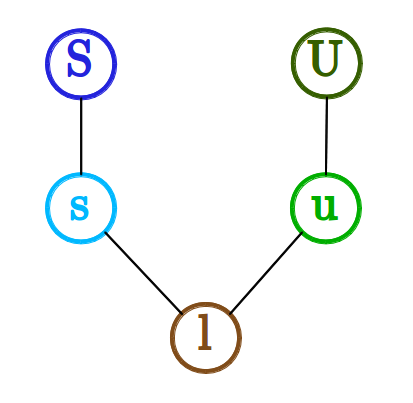}
    	\end{figure}	

The design matrix was first selected using BIC rather than AIC since we sought to explain axillary production rather than predict it. We compared the (sequential, logistic, complete$)_{\sigma}$ and (sequential, logistic, proportional$)_{\sigma}$ models for the six permutations $\sigma$: \{l, u, s, U, S\}, \{l, u, s, S, U\}, \{l, u, U, s, S\}, \{l, s, u, S, U\}, \{l, s, u, U, S\} and \{l, s, S, u, U\}. The complete design matrix was selected in all cases. We compared the six permuted (sequential, logistic, complete$)_{\sigma}$ models using the log-likelihood as criterion (see figure \ref{permuted_seq}). The third permutation $\sigma^*$ was the best, but the corresponding log-likelihood was very similar to the first two. Since models 1 and 2 (respectively 4 and 5) had exactly the same log-likelihood (illustrating Corollary \ref{seq_tau_comp}), they could not be differentiated. To differentiate all the permuted models we used a non-symmetric cdf $F$ (such as Gumbel min or Gumbel max), because Corollary \ref{seq_tau_comp} of invariance is no longer valid. The best result was obtained with the Gumbel max cdf, summarized in figure \ref{permuted_seq}. The third permutation $\sigma^*$ was still the best: \{l, u, U, s, S\}. Furthermore, a huge difference appeared between the first three permuted models and the last three. Therefore, the unspiny short shoot (u) seems to occur before the spiny short shoot (s).

	\begin{figure}[h!]
	\centering
	\caption{Log-likelihood of the (sequential, $F$, complete$)_{\sigma}$ models for the six permutations $\sigma$ for cdfs logistic (blue), Cauchy (green) and Gumbel max (magenta).}\label{permuted_seq}
    	\includegraphics[width=0.49\textwidth]{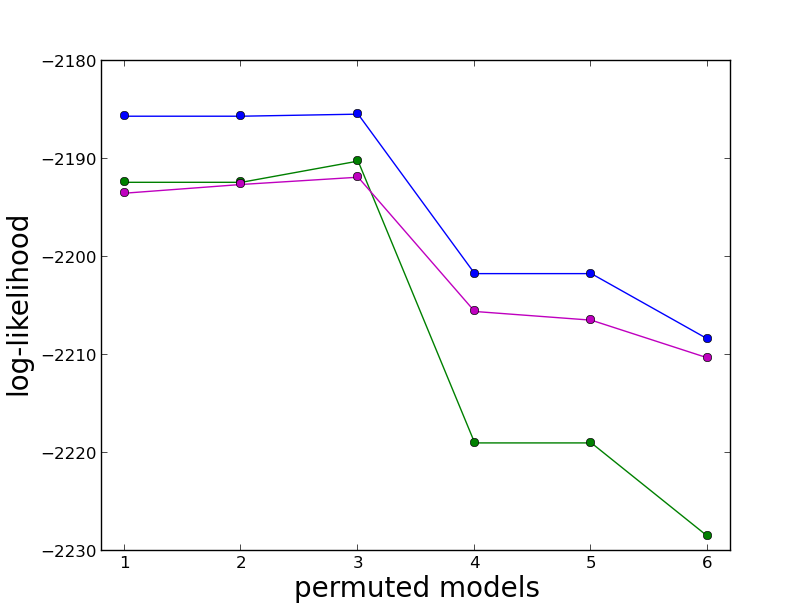}
    	\end{figure}

\section{\label{discussion} Concluding remarks}
The study of invariance properties showed that the systematic use of the logistic cdf is not justified. On the one hand, alternative cdfs allowed us to define the family of reference models for nominal data, for which the choice of the reference category is important. On the other hand, adjacent models are not appropriate for ordinal data with the logistic cdf.

For ordinal data, choosing the appropriate ratio is not obvious since adjacent, cumulative and sequential models have different properties and shortcomings. Adjacent models cannot be interpreted using latent variables. Cumulative models are not always defined (see section \ref{compatibility}). Finally, sequential models are not reversible. This might be explained by the different ordering interpretations. Sequential models correspond to non-reversible process ordering, while adjacent and cumulative models correspond to reversible scale ordering. 

The study of invariance properties allowed us to classify each $(r, F, Z)$ triplet as an ordinal or a nominal model. But which model is appropriate when categories are only partially ordered ? \citet{zhang} introduced the partitioned conditional model for partially-ordered set. They combine the multinomial logit model (for nominal data) with the odds proportional logit model (for ordinal data). More generally, we propose to combine any nominal model with any ordinal model \citep{peyhardi2013iwsm}. The unification proposed here help the specification of GLMs for partially ordered response variables.

\bibliographystyle{apalike}
\bibliography{biblio_phD}

\pagebreak

\appendix

\section*{Details on Fisher's scoring algorithm}\label{appendix_dpi_dr}
Below are details on computation of the Jacobian matrix $\frac{\partial \boldsymbol{\pi}}{\partial r}$ for four different ratios.
	\subsection*{Reference}
	For the reference ratio we have for $j=1,\ldots,J-1$
\begin{equation}
\pi_j = \frac{r_j}{1-r_j} \pi_J.	\label{ref_ratio}
\end{equation}
Summing on $j$ from $1$ to $J-1$ we obtain
\[ \pi_J = \frac{1}{1+ \sum_{j=1}^{J-1} \frac{r_j}{1-r_j} } . \]
The derivative of the product \eqref{ref_ratio} with respect to $r_i$ is
\begin{equation}
\dfdx{\pi_j}{r_i} = \dfdx{}{r_i} \left( \frac{r_j}{1-r_j} \right) \pi_J + \frac{r_j}{1-r_j} \dfdx{\pi_J}{r_i}.	\label{diff_ref}
\end{equation}
For the term of the sum part we obtain
\[ 
\dfdx{}{r_i} \left( \frac{r_j}{1-r_j} \right) = \left\{ 
\begin{array}{ll}
\displaystyle \frac{1}{(1-r_i)^2} & \mbox{if} \; i=j, \\
\displaystyle 0 & \mbox{otherwise}.
\end{array}
\right.
\]
For the second term of the sum we obtain
\[ \dfdx{\pi_J}{r_i} = - \frac{1}{(1-r_i)^2} \pi_J^2 .\]
Then \eqref{diff_ref} becomes
\[ 
\dfdx{\pi_j}{r_i} = \left\{ 
\begin{array}{ll}
\displaystyle  \frac{1}{r_i(1-r_i)} \left[ \frac{r_i}{1-r_i}\pi_J - \left( \frac{r_i}{1-r_i}\pi_J \right)^2 \right]   & \mbox{if} \; i=j, \\
\displaystyle  - \frac{1}{r_i(1-r_i)} \; \frac{r_i}{1-r_i} \; \frac{r_j}{1-r_j}\pi_J^2  & \mbox{otherwise}.
\end{array}
\right.
\]
Using \eqref{ref_ratio} again we obtain
\[ 
\dfdx{\pi_j}{r_i} = \frac{1}{r_i(1-r_i)} \left\{ 
\begin{array}{ll}
\displaystyle  \pi_i(1-\pi_i)  & \mbox{if} \; i=j, \\
\displaystyle  -\pi_i \pi_j  & \mbox{otherwise}.
\end{array}
\right.
\]
Finally we have
\[ \dfdx{\pi_j}{r_i} = \frac{\mbox{Cov}(Y_i,Y_j)}{F(\eta_i)[1-F(\eta_i)]}, \]
for row $i$ and column $j$ of the Jacobian matrix.

	\subsection*{Adjacent}
	For the adjacent ratio we have for $j=1,\ldots,J-1$
\[ \pi_j = \frac{r_j}{1-r_j} \pi_{j+1}.\]
and thus 
\begin{equation}
\pi_j = \left( \prod_{k=j}^{J-1}\frac{r_k}{1-r_k} \right) \pi_{J}.	\label{adj_ratio}
\end{equation}
Summing on $j$ from $1$ to $J-1$ we obtain
\[ \pi_J = \frac{1}{1+ \sum_{j=1}^{J-1} \prod_{k=j}^{J-1}\frac{r_k}{1-r_k} } . \]
The derivative of the product \eqref{adj_ratio} with respect to $r_i$ is
\begin{equation}
\dfdx{\pi_j}{r_i} = \dfdx{}{r_i} \left( \prod_{k=j}^{J-1}\frac{r_k}{1-r_k} \right) \pi_J + \left( \prod_{k=j}^{J-1}\frac{r_k}{1-r_k}\right) \dfdx{\pi_J}{r_i}.	\label{diff_adj}
\end{equation}
For the first term of the sum we obtain
\[ 
\dfdx{}{r_i} \left( \prod_{k=j}^{J-1}\frac{r_k}{1-r_k} \right) = \left\{ 
\begin{array}{ll}
\displaystyle \frac{1}{r_i(1-r_i)} \left( \prod_{k=j}^{J-1}\frac{r_k}{1-r_k} \right) & \mbox{if} \; i \geq j, \\
\displaystyle 0 & \mbox{otherwise}.
\end{array}
\right.
\]
For the second term of the sum we obtain
\[ \dfdx{\pi_J}{r_i} = - \frac{1}{r_i(1-r_i)} \left(\sum_{k=1}^i \prod_{k=j}^{J-1}\frac{r_k}{1-r_k} \right) \pi_J^2 .\]
Using \eqref{adj_ratio} it becomes
\[ \dfdx{\pi_J}{r_i} = - \frac{\pi_J}{r_i(1-r_i)} \sum_{k=1}^i \pi_k .\]
Then \eqref{diff_adj} becomes
\[ 
\dfdx{\pi_j}{r_i} = \left\{ 
\begin{array}{ll}
\displaystyle  \frac{1}{r_i(1-r_i)} \left( \pi_j - \pi_j \sum_{k=1}^i \pi_k \right)   & \mbox{if} \; i \geq j, \\
\displaystyle  - \frac{1}{r_i(1-r_i)} \pi_j \sum_{k=1}^i \pi_k  & \mbox{otherwise}.
\end{array}
\right.
\]
Finally we have
\[ 
\dfdx{\pi_j}{r_i} = \frac{1}{F(\eta_i)[1-F(\eta_i)]} \left\{ 
\begin{array}{ll}
\displaystyle  \pi_j(1-\gamma_i)  & \mbox{if} \; i \geq j, \\
\displaystyle  -\pi_j \gamma_i  & \mbox{otherwise},
\end{array}
\right.
\]
for row $i$ and column $j$ of the Jacobian matrix, where $\gamma_i=P(Y \leq i)=\sum_{k=1}^i \pi_k$.

	\subsection*{Cumulative}
	For the cumulative ratio we have for $j=1,\ldots,J-1$
\[ \pi_j = r_j - r_{j-1}, \]
with the convention $r_0=0$. Hence we obtain directly
\[ \frac{\partial \pi}{\partial r} = \begin{pmatrix}
1 & -1 & &   \\
  & 1  & \ddots &   \\
 &  & \ddots & -1 \\
 & &  & 1
\end{pmatrix}.
\]

	\subsection*{Sequential}
	For the sequential ratio we have for $j=1,\ldots,J-1$
\[ \pi_j = r_j \prod_{k=1}^{j-1} (1-r_k), \]
with the convention $\prod_{k=1}^{0} (1-r_k) =1$. Hence we obtain directly
\[\frac{\partial \pi_j}{\partial r_i} = \left\{
  \begin{array}{ll}
      \displaystyle  \prod_{k=1}^{j-1} \lbrace 1- F(\eta_{k}) \rbrace & \mbox{if } i=j, \\
      \displaystyle  -  F(\eta_j) \prod_{k=1, k\neq i}^{j-1} \lbrace 1- F(\eta_{k}) \rbrace & \mbox{if } i<j,\\   
        0 & \mbox{otherwise},
    \end{array}
\right. \]		
for row $i$ and column $j$ of the Jacobian matrix.	

\end{document}